\documentclass[aps,prl,nopacs,superscriptaddress,twocolumn,nobibnotes,nofootinbib]{revtex4-2}   

\usepackage{amsmath,amssymb,amsthm,bm,mathtools,amsfonts,mathrsfs,bbm,dsfont} %Useful math packages
\usepackage{graphicx} % Pac to insert external figures
\usepackage{xcolor,soul}
\usepackage[colorlinks,citecolor=blue,linkcolor=magenta,urlcolor=blue]{hyperref} 

\newcommand{\ket}[1]{|#1\rangle}
\newcommand{\bra}[1]{\langle #1 |}

\newcommand{\ketbra}[1]{| #1\rangle \langle #1|}
\newcommand{\id}{\ensuremath{\mathds{1}}}
\newcommand{\eins}{\ensuremath{\mathds{1}}}
\newcommand{\tr}{\ensuremath{{\rm{Tr}}}}
\newcommand{\mean}[1]{\ensuremath{ \left\langle #1 \right\rangle }}

\newcommand{\muell}[1]{}
\newcommand{\RR}{\mathcal{R}}

\newtheorem{theorem}{Theorem}
\newtheorem{corollary}[theorem]{Corollary}
\newtheorem{observation}[theorem]{Observation}

\newtheorem{lemma}[theorem]{Lemma}

\usepackage{cleveref}
\crefformat{footnote}{#2\footnotemark[#1]#3} % used for footnotes. Needed becquse of the thqnks command

\begin{document}

\title{Symmetries in quantum networks lead to no-go theorems for entanglement distribution and to verification techniques} 

\author{Kiara Hansenne}
\affiliation{Naturwissenschaftlich-Technische Fakultät, 
Universität Siegen, Walter-Flex-Straße 3, 57068 Siegen, Germany}

\author{Zhen-Peng Xu}
\thanks{\tt{zhen-peng.xu@uni-siegen.de}}
\affiliation{Naturwissenschaftlich-Technische Fakultät, 
Universität Siegen, Walter-Flex-Straße 3, 57068 Siegen, Germany}

\author{Tristan Kraft}
\affiliation{Institute for Theoretical Physics, University of Innsbruck, Technikerstraße 21A, 6020 Innsbruck, Austria}
\affiliation{Naturwissenschaftlich-Technische Fakultät, 
Universität Siegen, Walter-Flex-Straße 3, 57068 Siegen, Germany}

\author{Otfried Gühne}
\affiliation{Naturwissenschaftlich-Technische Fakultät, 
Universität Siegen, Walter-Flex-Straße 3, 57068 Siegen, Germany}

\date{\today}

\begin{abstract}
    Quantum networks are promising tools for the implementation of long-range quantum 
    communication. The characterization of quantum correlations in networks and their 
    usefulness for information processing is therefore central for the progress of the 
    field, but so far only results for small basic network structures or pure quantum 
    states are known. Here we show that symmetries provide a versatile tool for the 
    analysis of correlations in quantum networks. We provide an analytical approach 
    to characterize correlations in large network structures with arbitrary topologies. 
    As examples, we show that entangled quantum states with a bosonic or fermionic symmetry can not be generated in networks; moreover, cluster and graph states are not accessible. Our methods can be used to design certification methods for the functionality of specific links in a network and have implications for the design of future network structures.
\end{abstract}

\maketitle

\section{Introduction}
A central paradigm for quantum information processing is the notion 
of quantum networks~\cite{Kimble2008, Simon2017, Wehner2018, Biamonte2019}. 
In an abstract sense, a quantum network consists of quantum systems as nodes on specific locations, where some of the nodes are connected via links. These links correspond to quantum channels, which may be used to send quantum information (e.g., a polarized photon) or where entanglement may be distributed. Crucial building blocks for the links, such as photonic quantum channels between a satellite and ground stations~\cite{Yin2017, Liao2017, Liao2018, Yin2020} or the high-rate distribution of entanglement between nodes \cite{Humphreys2018, Stephenson2020} have recently been experimentally demonstrated. Clearly, such real physical implementations are always noisy and 
may only work probabilistically, but there are various theoretical
approaches to deal with this \cite{khatri,Omar21.1,Omar21.2, vanLoock}.

{For the further progress of the field,} it is essential to design methods 
for the certification and benchmarking of a given network structure or
a single specific link within it. In view of current experimental
limitations, the question arises which states can be prepared in the network 
with moderate effort, e.g., with simple local operations.
This question has attracted some attention, with several lines of research 
emerging. First, the problem has been considered in the classical setting, 
such as the analysis of causal structures {\cite{Chaves2014, Chaves2015, Wolfe2019}} 
or in the study of hidden variable models, where the hidden variables 
are not 
equally distributed between every party~{\cite{Branciard2010,Branciard2012, Fritz2012, Rosset2016, Renou2019, Gisin2020, deVicente21}}. Concerning quantum correlations, several initial 
works appeared in the last year, suggesting slightly different definitions of network entanglement~{\cite{Navascues2020,Kraft2020triangle,Aberg2020,Luo2020}}. These have been
further investigated~{\cite{Kraft2020cm,Luo2021, deVicente21a}} and methods from the classical realm  have been extended to the quantum scenario~{\cite{WolfePRX}}. Still, the present results are limited to simple networks like the triangle network, noise-free networks or networks build from specific quantum states, or 
bounded to small dimensions due to numerical limitations. 

{In this work,} we show an analytical approach to characterize quantum correlations 
in arbitrary network topologies. Our approach is based on symmetries, which may occur 
as permutation symmetries or invariance under certain local unitary transformations. 
Symmetries play an outstanding role in various {fields} of physics~{\cite{Coleman}} and they have already turned out to be useful for various other problems in quantum information theory~{\cite{Werner1989, Vollbrecht2001, Eckert2002, Caves2002, Doherty2002, Toth2009, Eltschka2012,Seevinck2009}}. On a technical level, we combine the inflation technique for quantum networks {\cite{Wolfe2019, Navascues2020}} with estimates known from the study of entropic uncertainty relations~{\cite{Toth2005, Niekamp2012, Wehner2008}}.
Based on our approach, we derive simply testable inequalities in terms of 
expectation values, which can be used to decide whether a given state may 
be prepared in a network or not. With this we can prove that large classes 
of states cannot be prepared in networks using simple communication, for 
instance all multiparticle graph states with up to twelve vertices with 
noise, as well as all mixed entangled permutationally symmetric states. 
This delivers various methods for benchmarking: First, the observation 
of such states in a network certifies the implementation of advanced 
network protocols. Second, our results allow to design simple tests 
for the proper working of a specific link in a given network.

%%%%%%%%%%%%%%%%%%%%%%%%%%%%%%%%%%%%%%%%%%%%%%%
\section*{Results}
%\vspace{2em}
%\noindent\textbf{\large Results}\\
%%%%%%%%%%%%%%%%%%%%%%%%%%%%%%%%%%%%%%%%%%%%%%%
%%%%%%%%%%%%%%%%%%%%%%%%%%%%%%%%%%%%%%%%%%%%%%%%%%%%%%%%%%%%%%%%%%%%%%%%
{\bf Network entanglement.}
%%%%%%%%%%%%%%%%%%%%%%%%%%%%%%%%%%%%%%%%%%%%%%%%%%%%%%%%%%%%%%%%%%%%%%%%%
To start, let us define the types of correlations that can be prepared in 
a network. In the simplest scenario Alice, Bob and Charlie aim to prepare 
a tripartite quantum state using three bipartite source states
$\varrho_a$, $\varrho_b$ and $\varrho_c$, see Fig.~\ref{fig:TriangleNetwork}. 
Parties belonging to a same source state are sent to different parties of the network, i.e.\@ $A$, $B$ or $C$, such that the global state reads $\varrho_{ABC} = \varrho_c \otimes \varrho_b \otimes \varrho_a.$ Note that here the order of the parties on both sides of the 
equation is different. After receiving the states, each party may still apply 
a local operation 
$\mathcal{E}_X$ (for $X=A,B,C$), in addition these operations may be 
coordinated by shared randomness. This leads to a global state of the form
\begin{equation} \label{eq:NteState}
    \varrho = 
    \sum_\lambda 
    p_\lambda 
    \mathcal{E}_A^{(\lambda)} \otimes \mathcal{E}_B^{(\lambda)} \otimes \mathcal{E}_C^{(\lambda)} \big[\varrho_{ABC}\big],
\end{equation}
and the question arises, which three-party states can be written in this form and which cannot?

%%%%%%%%%%%%%%%%%%%%%%%%%%%%%%%%%%%%%%%%%%%%%%%%%%%%
\begin{figure}
	\includegraphics[width=.75\linewidth]{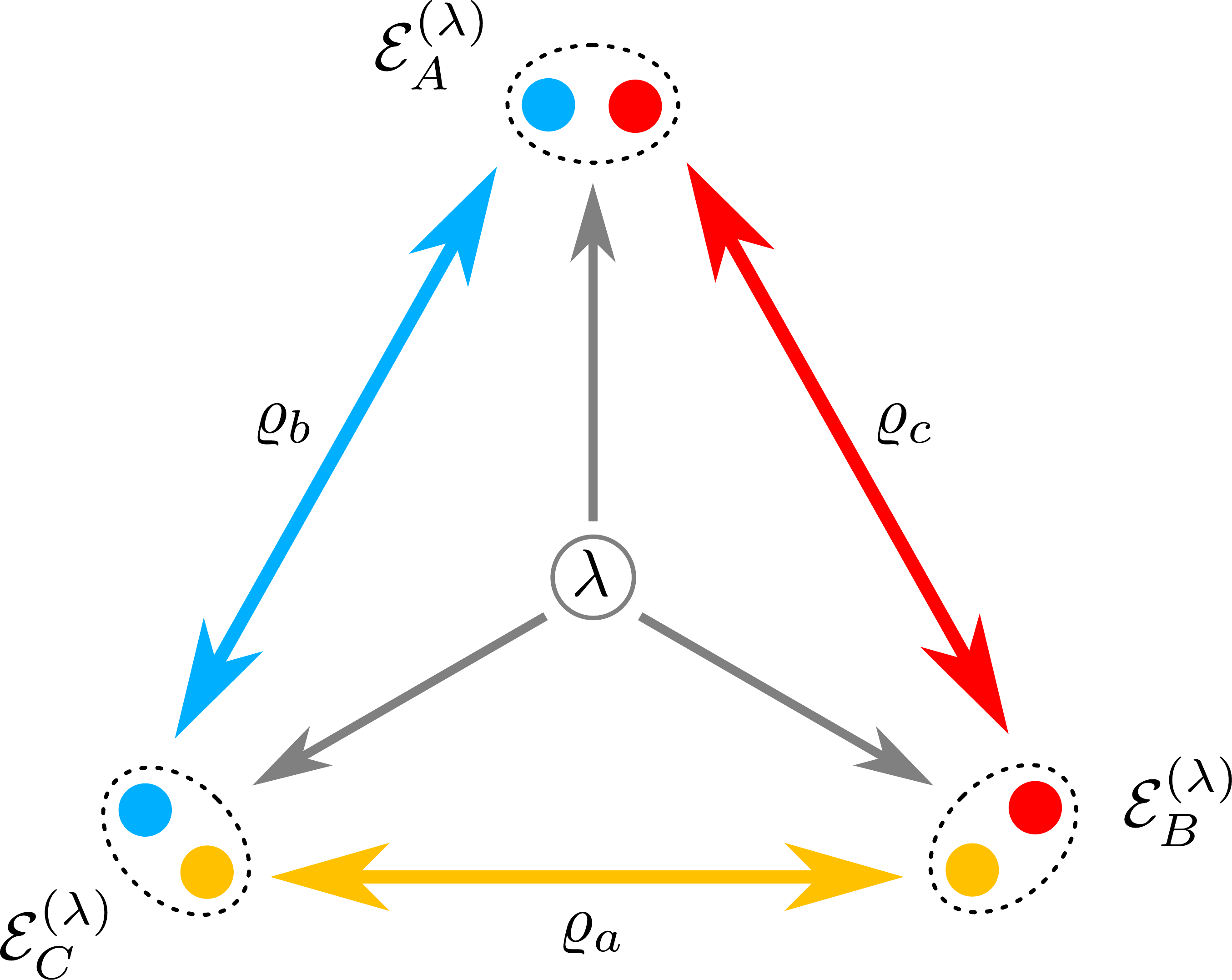}
	\caption{Triangle quantum network. Three sources $\varrho_a$, 
	$\varrho_b$ and $\varrho_c$ distribute parties to three nodes, Alice, 
	Bob and Charlie ($A$, $B$ and $C$). {The  colours yellow, blue and red are associated to the sources $\varrho_a$, $\varrho_b$ and $\varrho_c$ respectively.} Alice, Bob and Charlie each end 
	up with a bipartite system $X=X_1X_2$ on which they perform a local channel $\mathcal{E}_X^{(\lambda)}$ ($X=A,B,C$) depending on a 
	classical random variable $\lambda$.}
	\label{fig:TriangleNetwork}
\end{figure}
%%%%%%%%%%%%%%%%%%%%%%%%%%%%%%%%%%%%%%%%%%%%%%%%%%%

Some remarks are in order: First, the definition of network states 
in Eq.~(\ref{eq:NteState}) can directly be extended to more parties 
or more advanced sources, e.g.\@ one can consider the case of five 
parties $A,B,C,D,E,$ where some sources distribute four-party states 
between some of the parties. Second, the scenario considered uses 
local operations and shared randomness (LOSR) as allowed operations, 
which is a smaller set than local operations and classical communication 
(LOCC). In fact, LOCC are much more difficult to implement, but using 
LOCC and teleportation any tripartite state can be prepared 
from bipartite sources. On the other hand, the set LOSR is strictly 
larger than, e.g., the unitary operations considered {in  refs.}~\cite{Luo2020,Kraft2020triangle}. Finally, the discerning reader may have noticed 
that in Eq.~(\ref{eq:NteState}) the state $\varrho_{ABC}$ does not 
depend on the shared random variable $\lambda$, but since the 
dimension of the source states $\varrho_i$ is not bounded one 
can always remove a dependency on $\lambda$ in the $\varrho_i$
by enlarging the dimension~\cite{Navascues2020}. Equivalently, 
one may remove the dependency of the maps $\mathcal{E}_X$ on 
$\lambda$ and the shared randomness may be carried by the 
source states only. 

%%%%%%%%%%%%%%%%%%%%%%%%%%%%%%%%%%%%%%%%%%%%%%%%%%%%%%%%%%%
{\bf Symmetries. } 
%%%%%%%%%%%%%%%%%%%%%%%%%%%%%%%%%%%%%%%%%%%%%%%%%%%%%%%%%%%
Symmetry groups can act on quantum states in different ways. First, the 
elements of a unitary symmetry group may act transitively on the density 
matrix $\varrho$. That is, $\varrho$ is invariant under transformations 
like
\begin{equation} 
\label{eq:SymmDef}
\varrho \longmapsto U \varrho U^\dagger = \varrho.
\end{equation}
If $\varrho= \ketbra{\psi}$ is pure, this implies 
$U \ket{\psi} = e^{i \phi}\ket{\psi}$ and $\ket{\psi}$ is, up 
to some phase, an eigenstate of some operator. Second, for pure 
states one can also identify directly a certain subspace of the 
entire Hilbert space that is equipped with a certain symmetry, 
e.g.\@ symmetry under exchange of two particles. Denoting by 
$\Pi$ the projector onto this subspace, the symmetric pure states 
are defined via
\begin{equation}
\Pi \ket{\psi} = \ket{\psi}   
\label{eq:SymmDef2}
\end{equation}
and for mixed states one has $\varrho = \Pi \varrho \Pi$. Note that 
if $\varrho= \sum_k p_k \ketbra{\phi_k}$ has some decomposition 
into pure states, then each $\ket{\phi_k} = \Pi \ket{\phi_k}$ 
has to be symmetric, too. 

In the following, we consider mainly two types of symmetries. First, 
we consider multi-qubit states obeying a symmetry as in Eq.~(\ref{eq:SymmDef}) 
where the symmetry operations consist of an Abelian group of tensor products 
of Pauli matrices. These groups are usually referred to as stabilizers in 
quantum information theory~\cite{Gottesman}, and they play a central role in the 
construction of quantum error correcting codes. Pure states obeying such 
symmetries are also called stabilizer states, or, equivalently, graph states~\cite{Hein2006,Audenaert2005}. Second, we consider states with a permutational 
(or bosonic) symmetry \cite{Eckert2002, Toth2009,Kraus2003}, obeying relations as in Eq.~(\ref{eq:SymmDef2}) with 
$\Pi$ being the projector onto the symmetric subspace.

%%%%%%%%%%%%%%%%%%%%%%%%%%%%%%%%%%%%%%%%%%%%
{\bf GHZ states. }
%%%%%%%%%%%%%%%%%%%%%%%%%%%%%%%%%%%%%%%%%%%%%
As a warming-up exercise we discuss the Greenberger-Horne-Zeilinger (GHZ) 
state of three qubits, 
\begin{equation}
\ket{GHZ} = \frac{1}{\sqrt{2}} (\ket{000} + \ket{111})
\end{equation}
in the triangle scenario. This simple case was already the main example 
in previous works on network correlations~\cite{Navascues2020,Kraft2020triangle,Luo2020}, but it allows us to introduce our concepts and ideas in a simple setting, such that their 
full generalization is later conceivable. 

The GHZ state is an eigenstate of the observables
\begin{equation}
g_1 = X_A X_B X_C , \quad g_2 = 1_A Z_B Z_C , \quad g_3= Z_A Z_B 1_C.
\end{equation}
Here and in the following we use the shorthand notation
$1_A X_B Y_C = \eins\otimes\sigma_x \otimes \sigma_y$ for tensor
products of Pauli matrices. Indeed, these $g_k$ commute and generate 
the stabilizer $\mathcal{S} = \{\id, g_1, g_2, g_3, g_1g_2 , g_1g_3 , g_2g_3 , g_1g_2g_3 \}$.
Clearly, for any $S_i \in \mathcal{S}$ we have 
$S_i \ket{GHZ} = \ket{GHZ}$ and so $\bra{GHZ} S_i \ket{GHZ}=1$.

%%%%%%%%%%%%%%%%%%%%%%%%%%%%%%%%%%%%%%%%%%%%%%%%%%%%%%%%%%%%%%%%%%%%%%%%%%
\begin{figure}
	\includegraphics[width = 0.95\columnwidth]{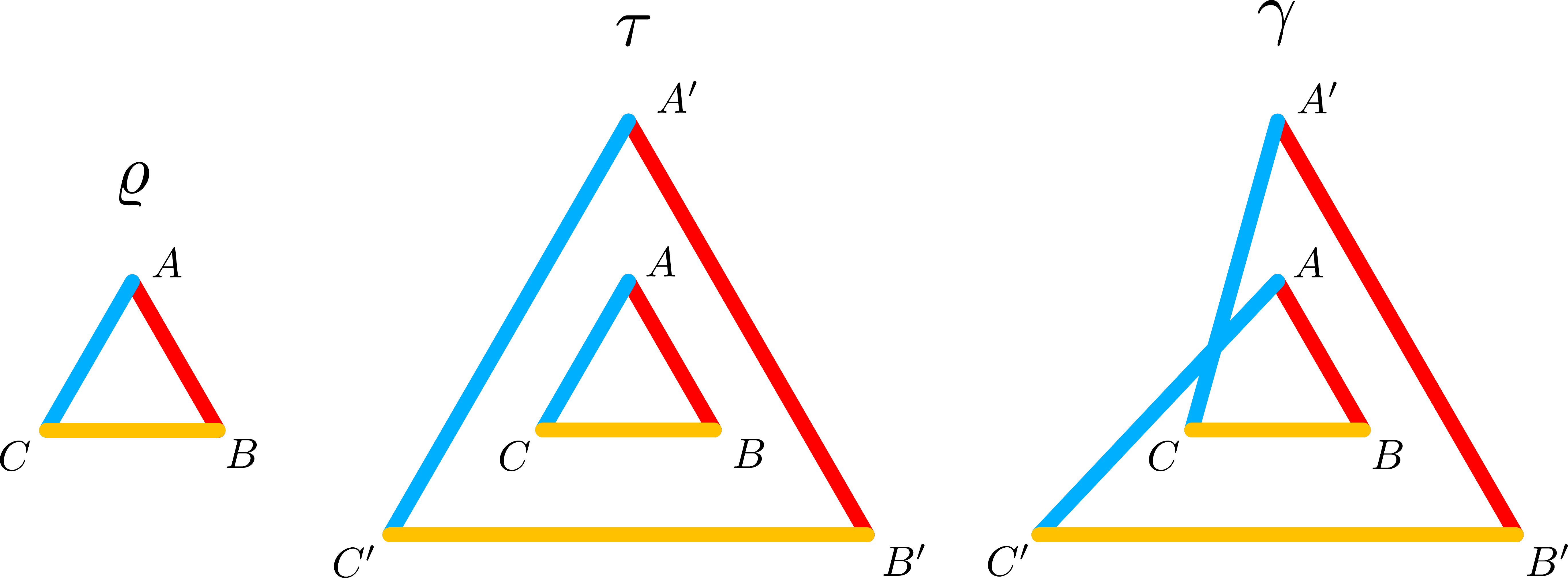}
	\caption{Triangle network and two of its inflations. The first figure represents the triangle network of Fig.\@ \ref{fig:TriangleNetwork}, with global 
	state $\varrho$ and parties $A$, $B$ and $C$. Using the same source states {(represented by lines of same 
	colour, i.e.\@ yellow, blue and red are associated to the sources $\varrho_a$, $\varrho_b$ and $\varrho_c$ of Fig.\@ {\ref{fig:TriangleNetwork}} respectively)} and same local channels, one can build the so-called inflated state 
	$\tau$ {with parties $X$, $X'$ ($X=A,B,C$)}. This state is separable with respect to the $ABC|A'B'C'$ partition. The 
	state $\gamma$ is build similarly, but with a rewiring of the sources, leading 
	to an inflated state that is in general not separable and different from 
	$\tau$. Still, this procedure imposes that several marginals of $\varrho$, 
	$\tau$ and $\gamma$ are equal, e.g.\@ $\varrho_{AC}= \tau_{A'C'} = \gamma_{A'C}$. {The parties of $\gamma$ are labeled in the same way than $\tau$. Note that this is a simplified version of Fig.\@ {\ref{fig:TriangleNetwork}}, i.e.\@ that the local channels and the randomness source are not depicted but implied.}
	}
	\label{fig:triangleandinflations}
\end{figure}
%%%%%%%%%%%%%%%%%%%%%%%%%%%%%%%%%%%%%%%%%%%%%%%%%%%%%%%%%%%%%%%%%%%

As a tool for studying network entanglement, we use the inflation
technique~\cite{Wolfe2019, WolfePRX}. The basic idea is depicted in Fig.~(\ref{fig:triangleandinflations}). If a state $\varrho$ can be 
prepared in the network scenario, then one can also consider a 
scenario where each source state is sent two-times to multiple 
copies of the parties. In this multicopy scenario, the source 
states may, however, also be wired in a different manner. In the 
simplest case of doubled sources, this may lead to two different 
states, $\tau$ and $\gamma$. Although $\varrho$, $\tau$ and $\gamma$ 
are different states, some of their marginals are identical, see Fig.~(\ref{fig:triangleandinflations}) and Supplementary Note 1. If one can 
prove that states $\tau$ and $\gamma$ with the marginal conditions
do not exist, then $\varrho$ cannot be prepared in the network. 

Let us start by considering the correlation  $\mean{Z_A Z_B}$
in $\varrho$, $\tau$ and $\gamma$. The values are equal in all 
three states, $\mean{Z_A Z_B}_\varrho = \mean{Z_A Z_B}_\tau =
\mean{Z_A Z_B}_\gamma$, and the same holds for the correlation 
$\mean{Z_B Z_C}.$ Note that these should be large, if $\varrho$
is close to a GHZ state, as $Z_A Z_B$ is an element of the 
stabilizer. Using the general relation $\mean{Z_A Z_C} \geq \mean{Z_A Z_B}
+ \mean{Z_B Z_C} - 1 $~\cite{Navascues2020} we can use this to estimate 
$\mean{Z_A Z_C}$ in $\gamma$. Due to the marginal conditions, 
we have $\mean{Z_A Z_C}_\gamma = \mean{Z_{A'} Z_C}_\tau$, implying
that this correlation in $\tau$ must be large, if $\varrho$ is
close to a GHZ state. On the other hand, the correlation $\mean{X_A X_B X_C}$ 
corresponds also to a stabilizer element and should be large in the 
state $\varrho$ as well as in $\tau$.

The key observation is that the observables $X_A X_B X_C$ and $Z_{A'} Z_C$
anticommute; moreover, they have only eigenvalues $\pm 1$. For this 
situation, strong constraints on the expectation values are known: 
If $M_i$ are pairwise anticommuting observables with eigenvalues $\pm 1$, 
then $\sum_i \mean{M_i}^2 \leq 1$~\cite{Toth2005}. This fact has already
been used to derive entropic uncertainty relations~\cite{Wehner2008,Niekamp2012}
or monogamy relations~\cite{Tran2018,Kurzynski2011}. For our situation, it directly
implies that for $\tau$ the correlations $\mean{X_A X_B X_C}$  and 
$\mean{Z_{A'} Z_C}$ cannot both be large. Or, expressing everything
in terms of the original state $\varrho$, if 
$\mean{Z_A Z_B} + \mean{Z_B Z_C} - 1 \geq 0$ 
then a condition for preparability of a state in the network is
\begin{equation}
\mean{X_A X_B X_C}^2 + (\mean{Z_A Z_B} + \mean{Z_B Z_C} - 1)^2 \leq 1.
\label{eq-ghz-condition}
\end{equation}
This is clearly violated by the GHZ state. In fact, if one considers
a GHZ state mixed with white noise, $\varrho = p \ketbra{GHZ} 
+ (1-p) \eins/8$, then these states are detected already for 
$p > 4/5.$ Note that using the other observables of the stabilizer
and permutations of the particles, also other conditions like 
$\mean{Y_A Y_B X_C}^2 + (\mean{Z_A Z_B} + \mean{Z_A Z_C} - 1)^2 \leq 1$ 
can be derived. 

Using these techniques as well as concepts based on covariance
matrices~\cite{Aberg2020,Kraft2020cm} and classical networks \cite{Gisin2020}, 
one can also derive bounds on the maximal GHZ fidelity achievable by network 
states. In fact, for network states
\begin{equation}
F_{GHZ} = \bra{GHZ}\varrho\ket{GHZ} \leq \frac{1}{\sqrt{2}} \approx 0.7071
\end{equation}
holds, as explained in Supplementary Note 2. This is a clear improvement on previous 
analytical bounds, although it does not improve a fidelity bound obtained 
by numerical convex optimization~\cite{Navascues2020}.

%%%%%%%%%%%%%%%%%%%%%%%%%%%%%%%%%%%%%%%%%%%%%%%%%%%%%%%
{\bf Cluster and graph states. }
%%%%%%%%%%%%%%%%%%%%%%%%%%%%%%%%%%%%%%%%%%%%%%%%%%%%%%%
The core advantage of our approach is the fact that it can directly 
be generalized to more {parties} and complicated networks, while the 
existing numerical and analytical approaches are mostly restricted 
to the triangle scenario. 

%%%%%%%%%%%%%%%%%%%%%%%%%%%%%%%%%%%%%%%%%%%%%%%%%%%%%%%%%%%%%%%5
\begin{figure}[t]
    \includegraphics[width = 0.95\columnwidth]{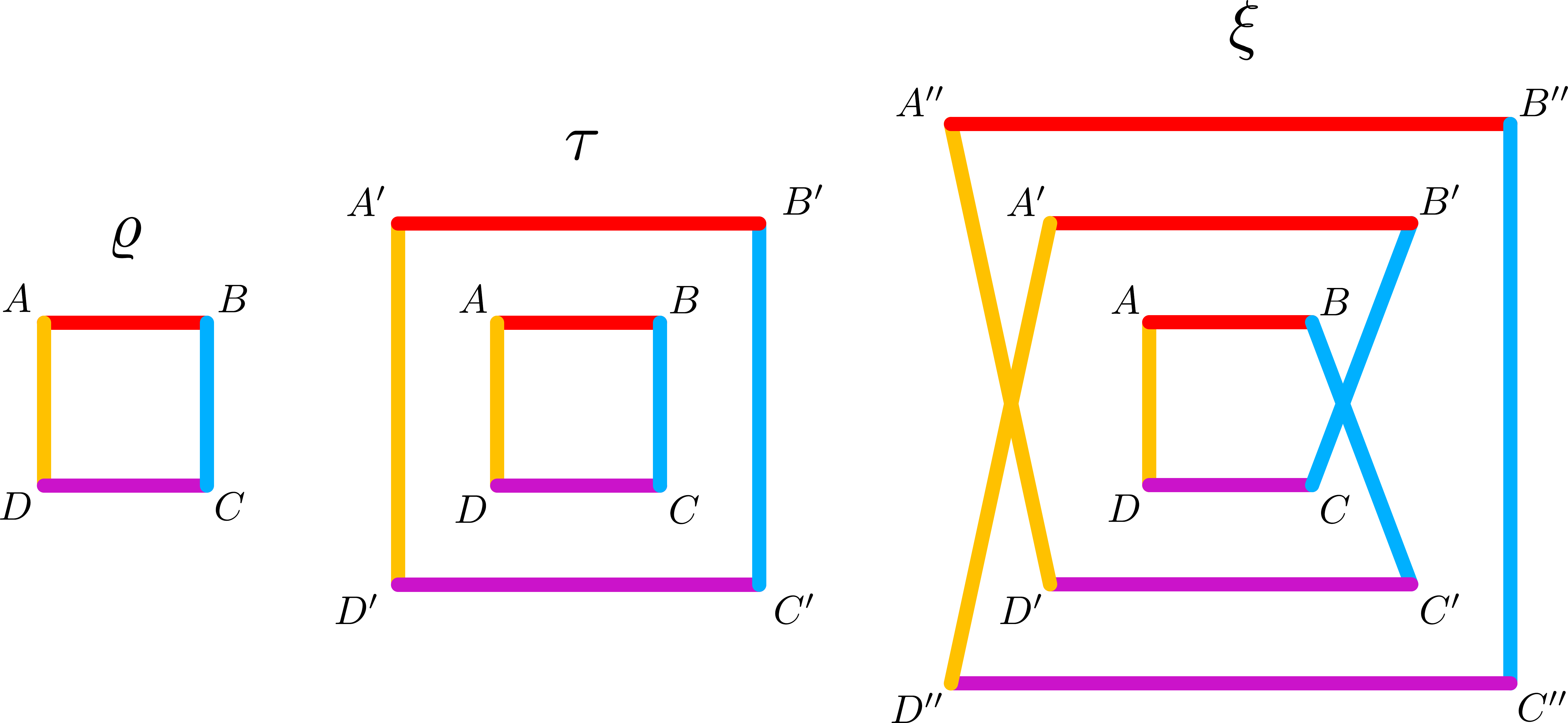}
    \caption{Square network and two of its inflations. Similar 
    to the triangle network of Fig.\@ \ref{fig:triangleandinflations}, 
    the state $\tau$ is generated using two copies of the sources and 
    channels used to generate $\varrho$. Then, one goes to a higher-order 
    inflation by using three copies of the sources and the local channels. 
    By rewiring one obtains the inflated state $\xi$. Again, one has several 
    equalities between the marginals of $\varrho$, $\tau$ and $\xi$. {Parties with a link of same colour are connected by an identical source, as in Fig.\@ {\ref{fig:triangleandinflations}}. The labels $X$, $X'$, $X''$ ($X=A,B,C$) denote the parties of the states.}
    }
	\label{fig:square}
\end{figure}
%%%%%%%%%%%%%%%%%%%%%%%%%%%%%%%%%%%%%%%%%%%%%%%%%%%%%%%%%%%%%%%%%%%%

Let us start the discussion with the four-qubit cluster state 
$\ket{C_4}.$ This may be defined as the unique common $+1$-eigenstate of 
\begin{align}
g_1 &= X_A Z_B 1_C Z_D, \ g_2 = Z_AX_BZ_C1_D, \nonumber \\
g_3 &=1_AZ_BX_CZ_D, \ 
g_4 = Z_A1_BZ_CX_D.
\label{eq-cluster-stab}
\end{align}
For later generalization, it is useful to note that the choice of 
these observables is motivated by a graphical analogy. For the 
square graph in Fig.~\ref{fig:square} one can associate to any 
vertex a stabilizing operator in the following manner: One
takes $X$ on the vertex $i$, and $Z$ on its neighbours, i.e. 
the vertices connected to $i$. This delivers the observables in 
Eq.~(\ref{eq-cluster-stab}), but it may also be applied to 
general graphs, leading to the notion of graph states~\cite{Hein2006}.

If a quantum state can be prepared in the square network, 
then we can consider the third order
inflated state $\xi$ shown in Fig.~\ref{fig:square}. In the 
inflation $\xi$, the three observables $X_{B''} X_D$, $Z_{B'} X_C Z_D$, 
and $ X_{A} Y_{B} Y_D$ anticommute. These observables act
on marginals {that} are identical to those in $\varrho$. 
Consequently, for any state that can be prepared in the 
square network, the relation
\begin{equation} 
\label{eq:anti-commI3}
\mean{X_B X_D}^2 + \mean{Z_{B} X_C Z_D}^2 + \mean{X_{A} Y_{B} Y_D}^2 \leq 1
\end{equation}
holds. All these observables are also within the stabilizer of 
the cluster state, so the cluster state violates this inequality 
with an lhs equal to three. This proves that cluster states mixed 
with white noise cannot be prepared in the square network for 
$p > 1/\sqrt{3} \approx 0.577.$ Again, with the same strategy 
different nonlinear witnesses like 
$\mean{X_A X_C}^2 + \mean{Y_A Y_B Z_C Z_D}^2 \leq 1$ for network 
entanglement in the square network can be derived. This follows
from the second inflation in Fig.~\ref{fig:square}. Furthermore, it can be shown that states 
with a cluster state fidelity of $F_{C_4} > 0.7377$ cannot be 
prepared in a network. The full discussion is given in 
Supplementary Note 3.

This approach can be generalized to graph states. As already 
mentioned, starting from a general graph one can define a graph 
state by stabilizing operators in analogy to Eq.~(\ref{eq-cluster-stab}).
The resulting states play an eminent role in quantum information 
processing. For instance, the so-called cluster states, which correspond 
to graphs of quadratic and cubic square lattices, are resource states 
for measurement-based quantum computation~\cite{Raussendorf2001} and topological error 
correction~\cite{Kitaev2003,Yao2012}.

Applying the presented ideas to general graphs results in the following: 
If a graph contains a triangle, then under simple and weak conditions
an inequality similar to Eq.~(\ref{eq-ghz-condition}) can be derived. {This then
excludes} the preparability of noisy graph states in any
network with bipartite sources only. Note that this is a stronger 
statement than proving network entanglement only for the network 
corresponding to the graph, as was considered above for the cluster
state. 

At first sight, the identification of a specific triangle in the 
graph may seem a weak condition, but here the entanglement theory
of graph states helps: It is well known that certain transformations
of the graph, so-called local-complementations, change the graph state
only by a local unitary transformation~\cite{Nest2004,Hein2004}, so one may apply these to 
generate the triangle with the required properties. Indeed, this works for
all cases we considered (e.g.\@ the full classification up to twelve qubits 
from Refs.~\cite{Adan2009,Cabello2011, Danielsen2011}) and we can summarize: 

\noindent
{\bf{Observation 1.}} {\it (a) No graph state with up to twelve vertices can be prepared in a network with only bipartite sources. (b) If a graph contains a vertex with degree $d \leq 3$, then it cannot be prepared in any network with bipartite sources. (c) The two- and three-dimensional cluster states cannot be prepared in any network.}

In all the cases, it follows that graph states mixed with white noise, 
$\varrho = p \ketbra{G} + (1-p) \eins /2^N$ are network entangled for 
$p>4/5$, independently of the number of qubits. A detailed discussion 
is given in Supplementary Note 3.

{As mentioned above, the exclusion of noisy graph states from the set of network states with bipartite sources holds for all graph states we considered. Therefore, we conjecture that this is valid for all graph states, without restrictions on the number of parties.}

{We note that similar statements on entangled multiparticle states and symmetric 
states were made in Ref.~{\cite{Luo2020}}. However, we stress that the methods to obtain these results are very different from the anticommuting method used here, and that Observation 1 is only an application of this method (see section on the certification of network links for another use). Furthermore, the result of Ref.~{\cite{Luo2020}} concerning permutationally symmetric states only holds for pure states, whereas in the next section we will see that it holds for all permutationally symmetric states.}

{
A natural questions that arises is whether this method might be useful to characterize correlations in networks with more-than-bipartite sources. While this still needs to be investigated in details, examples show that the answer is most likely positive: Using the anticommuting relations, we demonstrate in Supplementary Note 3 that some states cannot be generated in networks with tripartite sources.
}

%%%%%%%%%%%%%%%%%%%%%%%%%%%%%%%%%%%%%%%%%%%%%%%%%%%%%%%%
{\bf Permutational symmetry. }
%%%%%%%%%%%%%%%%%%%%%%%%%%%%%%%%%%%%%%%%%%%%%%%%%%%%%%%%
Now we consider multiparticle quantum states of arbitrary dimension that obey a permutational or bosonic symmetry. Mathematically, these states act on the symmetric subspace only, meaning that $\varrho = \Pi^+ \varrho \Pi^+$, where $\Pi^+$ is the projector on the symmetric subspace. For example, in the case of three qubits this 
space is four-dimensional, and spanned by the Dicke states
$
\ket{D_0} = \ket{000}$, 
$\ket{D_1} = (\ket{001}+\ket{010}+\ket{100})/\sqrt{3}$, 
$\ket{D_2} = (\ket{011}+\ket{101}+\ket{110})/\sqrt{3}$, 
and 
$\ket{D_3} = \ket{111}$.

The symmetry has several consequences \cite{Eckert2002, Toth2009}. First, if one has a decomposition $\varrho= \sum_k p_k \ketbra{\psi_k}$ into pure states, then 
all $\ket{\psi_k}$ have to come from the symmetric space {} too. Since pure
symmetric states {} are either fully separable (like 
$\ket{D_0}$) or genuine multiparticle entangled (like
$\ket{D_1}$), this implies that mixed symmetric states 
have also only these two possibilities. That is, if a mixed symmetric 
state is separable for one bipartition, it must be fully separable. 

Second, permutational invariance can also be characterized
by the flip operator $F_{XY} = \sum_{ij} \ket{ij}_{XY}\bra{ji}_{XY}$ on
the particles $X$ and $Y$. Symmetric multiparticle states obey 
$\varrho = F_{XY}\varrho = \varrho F_{XY}$ for any pair
of particles, where the second equality directly follows
from hermiticity. Conversely, concluding  full permutational
symmetry from two-particle properties {only requires}
this relation for pairs such that the $F_{XY}$ generate
the full permutation group. Finally, it is easy to check
that if the marginal $\varrho_{XY}$ of a multiparticle state
$\varrho$ obeys $\varrho_{XY}= F_{XY}\varrho_{XY}$, then 
the full state $\varrho$ obeys the same relation, too. 

Armed with these insights, we can explain the idea for our main
result. Consider a three-particle state with bosonic symmetry
{that} can be prepared in a triangle network, the inflation 
$\gamma$ from Fig.~(\ref{fig:triangleandinflations}) and 
the reduced state $\gamma_{ABC}$ in this inflation. This 
obeys $\gamma_{ABC} = F_{XY}\gamma_{ABC}$ for $XY$ equal 
to $AB$ or $BC$. Since $F_{AB}F_{BC}F_{AB}=F_{AC}$, this 
implies that the reduced state $\tau_{AC'}$ obeys 
$\tau_{AC'}=F_{AC'}\tau_{AC'}$ and hence also the six-particle 
state $\tau.$ Moreover, $\tau$ also obeys similar constraints 
for other pairs of particles (like $AB$, $BC$, $A'B'$ and $B'C'$) 
and {it is} easy to see that jointly with $AC'$ these generate the 
full permutation group. So, $\tau$ must be fully symmetric. But 
$\tau$ is separable with respect to the $ABC|A'B'C'$ bipartition, 
so $\tau$ and hence $\varrho = \tau_{ABC}$ must be fully separable. 

The same argument can easily be extended to more complex networks, which
are not restricted to use bipartite sources and {holds for states of arbitrary local dimension}. We can summarize:

\noindent
{\bf Observation 2.} {\it Consider a permutationally symmetric state of $N$ parties. This state can be generated by a
network with $(N-1)$-partite sources if and only if it is fully separable.}

 We add that this Observation can also be extended to the case of fermionic
 antisymmetry, a detailed discussion is given in the Supplementary Note 4.

%%%%%%%%%%%%%%%%%%%%%%%%%%%%%%%%%%%%%%%%%%%%%%%%%%%%%%%%%%5
{\bf Certifying network links. }
%%%%%%%%%%%%%%%%%%%%%%%%%%%%%%%%%%%%%%%%%%%%%%%%%%%%%%%%%%%
For the technological implementation of quantum networks, it is 
of utmost importance to design certification methods to test
and benchmark different realizations. One of the basic questions
is, whether a predefined quantum link works or not. Consider a network where the link between two particles is absent 
or not properly working. For definiteness, we may consider the 
square network on the lhs of Fig.~\ref{fig:square} and the 
parties $A$ and $C$. In the second inflation $\tau$ we have 
for the marginals $\tau_{AC} = \tau_{A'C}$. This implies that 
the observables $X_A X_C$ and $Z_A Z_C$ on the original
state $\varrho$ correspond to anticommuting observables
on $\tau$, so we have $\mean{X_A X_C}^2 + \mean{Z_A Z_C}^2 
\leq 1$. Using higher-order inflations, one can extend and formulate
it for general networks: If a state can be prepared in a network with
bipartite sources but without the link $AC$, then
\begin{equation}
\mean{X_A X_C P_{R_1}}^2 
+ 
\mean{Y_A Y_C P_{R_2}}^2 
+ 
\mean{Z_A Z_C P_{R_3}}^2 
\leq 1.
\end{equation}
Here the $P_{R_i}$ are arbitrary observables on disjoint subsets of the other
particles, $R_i \cap R_j = \emptyset$. If a state was indeed  prepared in a real 
quantum network then violation of this inequality proves that the link $AC$ is 
working and distributing entanglement. In Supplementary Note 5, details are discussed and
examples are given, where this test allows to certify the functionality of
a link even if the reduced state $\varrho_{AC}$ is separable.

%%%%%%%%%%%%%%%%%%%%%%%%%%%%%%%%%%%%%%%
\section*{Discussion}
%%%%%%%%%%%%%%%%%%%%%%%%%%%%%%%%%%%%%%
We have provided an analytical method to analyze correlations arising in 
quantum networks from few measurements. With this, we have shown that 
large classes of states with symmetries, namely noisy graph states and permutationally symmetric states cannot be prepared in networks. Moreover, our approach allows to design simple tests for the functionality of a specified link in a network.

Our results open several research lines of interest. First, they are 
of direct use to analyze quantum correlations in experiments and to show that 
multiparticle entanglement is needed to generate observed quantum correlations. 
Second, they are useful for the design of networks in the realistic setting: 
For instance, we have shown that the generation of graph states from bipartite 
sources necessarily requires at least some communication between the parties, which may be of relevance for quantum repeater schemes based on graph states that have been designed~\cite{Epping2016}. Moreover, it has been shown that GHZ states provide an advantage for multipartite conference key agreement over bipartite sources~\cite{Epping2017}, which may be directly connected to the fact that their symmetric entanglement is inaccessible in networks. Finally, our {results} open the door for further studies of entanglement in networks, e.g. using limited communication (first results on this have recently been reported \cite{Spee2021}) or restricted quantum memories, which is central for future
realizations of a quantum internet.

\vspace{-0.4cm}
%%%%%%%%%%%%%%%%%%%%%%%%%%%%%%%%%%%%%%%%%%%%%%%%%%%%
\subsection{Acknowledgments}
%\vspace{2em}
%\noindent\textbf{\large Acknowledgments}\\
%%%%%%%%%%%%%%%%%%%%%%%%%%%%%%%%%%%%%%%%%%%%%%%%%%%%
We thank Xiao-Dong Yu and Carlos de Gois for discussions. This work was supported by the 
Deutsche Forschungsgemeinschaft (DFG, German Research Foundation, project numbers 447948357 and 440958198), the Sino-German Center for Research Promotion (Project M-0294), and the ERC (Consolidator Grant 683107/TempoQ). K.H. acknowledges support from the House of Young Talents of the University of Siegen. Z.P.X. acknowledges support from the Humboldt foundation. {T.K. acknowledges support from the Austrian Science Fund (FWF): P 32273-N27.}

%%%%%%%%%%%%%%%%%%%%%%%%%%%%%%%%%%%%%%%%%%%%%%%%%%%%
\vspace{-0.4cm}
\subsection{Author contributions}
%\vspace{2em}
%\noindent\textbf{\large Author contributions}\\
%%%%%%%%%%%%%%%%%%%%%%%%%%%%%%%%%%%%%%%%%%%%%%%%%%%%
K.H., Z.P.X., T.K. and O.G. derived the results and wrote the manuscript. K.H. and Z.P.X. contributed equally to the project. O.G. supervised the project. Correspondence and requests for materials should be addressed to Z.P.X.

\newpage
\onecolumngrid

%%%%%%%%%%%%%%%%%%%%%%%%%%%%%%%%%%%%%%%%%%%%%%%%%%%%%%%%%%%%%%%%%%%%%%%%%
\section{Supplementary Note 1: Network Correlations and the Inflation Technique}
%%%%%%%%%%%%%%%%%%%%%%%%%%%%%%%%%%%%%%%%%%%%%%%%%%%%%%%%%%%%%%%%%%%%%%%%%

Before explaining the inflation technique in some detail, it is useful 
to note some basic observations on the definition of network correlations.
Recall from the main text that triangle network states are of the form
\begin{equation}
\label{eq:NteStateAPPENDIX}
\varrho = 
\sum_\lambda p_\lambda 
\mathcal{E}_A^{(\lambda)} \otimes \mathcal{E}_B^{(\lambda)} \otimes \mathcal{E}_C^{(\lambda)} \big[\varrho_{ABC}\big],
\end{equation}
i.e.\@ there exist source states $\varrho_a$, $\varrho_b$, $\varrho_c$ with 
$\varrho_{ABC} =\varrho_a \otimes \varrho_b \otimes \varrho_c$, a shared random 
variable $\lambda$ and channels (that is, trace preserving positive maps) 
$\mathcal{E}_A^{(\lambda)}$, $\mathcal{E}_B^{(\lambda)}$ and $\mathcal{E}_C^{(\lambda)}$ 
that can be used to generate the state $\varrho$, as shown in Supplementary Fig.\@ \ref{fig:TriangleNetwork}.

%%%%%%%%%%%%%%%%%%%%%%%%%%%%%%%%%%%%%%%%%%%%%%%%%%%%%%%%%%%%%%%%%%%%%%%%%%%%%%%%%%%%%%%%%%%%%%%%%%%
\begin{figure}[b]
	\includegraphics[width=.4\linewidth]{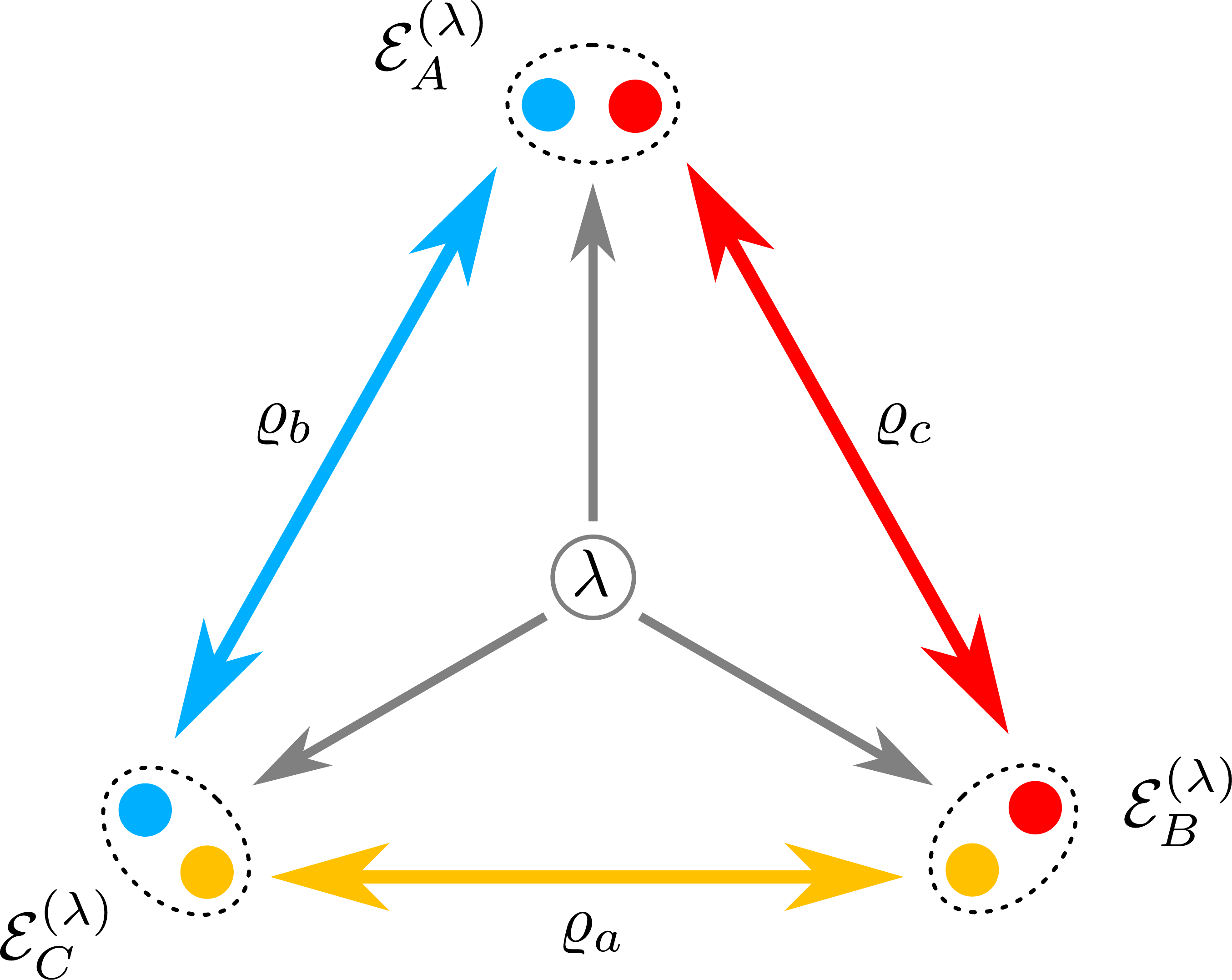}
	\caption{\textit{Triangle quantum network.} Three sources $\varrho_a$, 
	$\varrho_b$ and $\varrho_c$ distribute parties to three nodes, Alice, 
	Bob and Charlie ($A$, $B$ and $C$). Alice, Bob and Charlie each end 
	up with a bipartite system $X=X_1X_2$ on which they perform a local 
	channel $\mathcal{E}_X^{(\lambda)}$ ($X=A,B,C$) depending on a 
	classical random variable $\lambda$.}
	\label{fig:TriangleNetwork}
\end{figure}
%%%%%%%%%%%%%%%%%%%%%%%%%%%%%%%%%%%%%%%%%%%%%%%%%%%%%%%%%%%%%%%%%%%%%%%%%%%%%%%%%%%%%%%%%%%%%%%%%%%

First, we note that in this definition, the state $\varrho_{ABC}$ does not depend on 
the classical variable $\lambda$. This is however, no restriction, as the dimensions 
of the source states are not bounded. If in Eq.~(\ref{eq:NteStateAPPENDIX}) the 
$\varrho_{ABC}(\lambda)$ and hence the $\varrho_a(\lambda)$, $\varrho_b(\lambda)$, $\varrho_c(\lambda)$ depend on $\lambda$, one can just combine the set of all 
$\varrho_a(\lambda)$ to a single $\varrho_a$ etc.~and redefine the maps 
$\mathcal{E}_A^{(\lambda)}$ etc.~such that they act on the appropriate 
$\varrho_a(\lambda)$. This results in a form where $\varrho_{ABC}$
does not depend on $\lambda$ anymore, hence the state can be written as in
Eq.~(\ref{eq:NteStateAPPENDIX}). Note that this has already been observed 
in~\cite{Navascues2020}. 

As mentioned in the main text, one may define network 
states also in a manner where the shared randomness is carried by the sources 
only. Indeed, if one adds an ancilla system to the source states, this may 
be used to identify the channel $\mathcal{E}_X$ to be applied. More explicitly, 
the source states may be redefined as $\varrho_c^{(\lambda)} \otimes \ketbra{\lambda}$ with orthogonal ancilla states $\ket{\lambda}$ being 
send to Bob (respectively to Charlie and Alice for sources $a$ and $b$),
such that Bob can, by measuring $\ket{\lambda}$, decide which channel to 
apply. This measurement can then be seen as a global channel $\mathcal{E}_B$
that does not depend on $\lambda$.
%the global channel applied on Alice's party would be $\mathcal{E}_A [\varrho_A^{(\lambda)}] = \sum_\kappa \tr\big((\varrho_A^{(\lambda)})_{\text{ancilla}} \ketbra\kappa\big) \ \mathcal{E}_A^{(\kappa)} \big[ \tr_{\text{ancilla}}(\varrho_A^{(\lambda)}) \big] $ etc. 
From the linearity of the maps, one may also write general network states
as $\varrho = \mathcal{E}_A \otimes \mathcal{E}_B \otimes \mathcal{E}_C 
\big[ \sum_\lambda p_\lambda \varrho_{ABC}^{(\lambda)}\big]$ as an equivalent definition. For our purpose, the potential dependence of $\varrho_{ABC}$ on 
$\lambda$ has the following consequence: If we wish to compute for a symmetric $\ket{\psi}$ the maximum fidelity $\bra{\psi} \varrho \ket{\psi}$ over all network states $\varrho$, then we may assume that $\varrho$ permutationally symmetric, too. This follows from the simple fact that we can, without decreasing the overlap, 
symmetrize the state $\varrho_{ABC},$ and the symmetrized state will still
be preparable in the network.

Second, one may restrict the $\varrho_{ABC} =\varrho_a \otimes \varrho_b \otimes \varrho_c$
further. Indeed it is straightforward to see that $\varrho_a=\ketbra{a}$,
$\varrho_b=\ketbra{b}$, and $\varrho_c=\ketbra{c}$ can be chosen to be pure, 
as the channels $\mathcal{E}_A^{(\lambda)}$ etc.~are linear.

Third, as the set of network preparable states is by definition convex, one may 
ask for its extremal points. Formally, these are of the form 
$\mathcal{E}_A^{(\lambda)} \otimes \mathcal{E}_B^{(\lambda)} \otimes \mathcal{E}_C^{(\lambda)} \big[\ketbra{a} \otimes \ketbra{b}\otimes \ketbra{c}\big],$ but can these further 
be characterized? Clearly, pure biseparable three-particle states, such as
$\ket{\psi}=\ket{\phi}_{AB}\otimes\ket{\eta}_C$ are extremal points. There are 
however, also mixed states as extremal points, which can be seen as follows: It 
was shown in Ref.~\cite{Luo2020} that pure three-qubit states which are genuine multiparticle 
entangled (that is, not biseparable) cannot be prepared in the triangle network. 
On the other hand, in Ref.~\cite{Navascues2020} it was shown that there are network states
having a GHZ fidelity of 0.5170, which implies that they are genuine multiparticle
entangled \cite{Seevinck2009}. So, the set defined in Eq.~(\ref{eq:NteStateAPPENDIX})
must have some extremal points, which are genuine multiparticle entangled mixed states.

After this prelude, let us explain the inflation technique~\cite{WolfePRX, Wolfe2019}, which has already proven to be useful for the characterization of quantum
networks~\cite{Navascues2020}. We introduce it for triangle networks as 
for arbitrary networks it is a direct generalization.

%%%%%%%%%%%%%%%%%%%%%%%%%%%%%%%%%%%%%%%%%%%%%%%%%%%%%%%%%%%%%%%%%%%%%%%
\begin{figure} 
	\includegraphics[scale = .15]{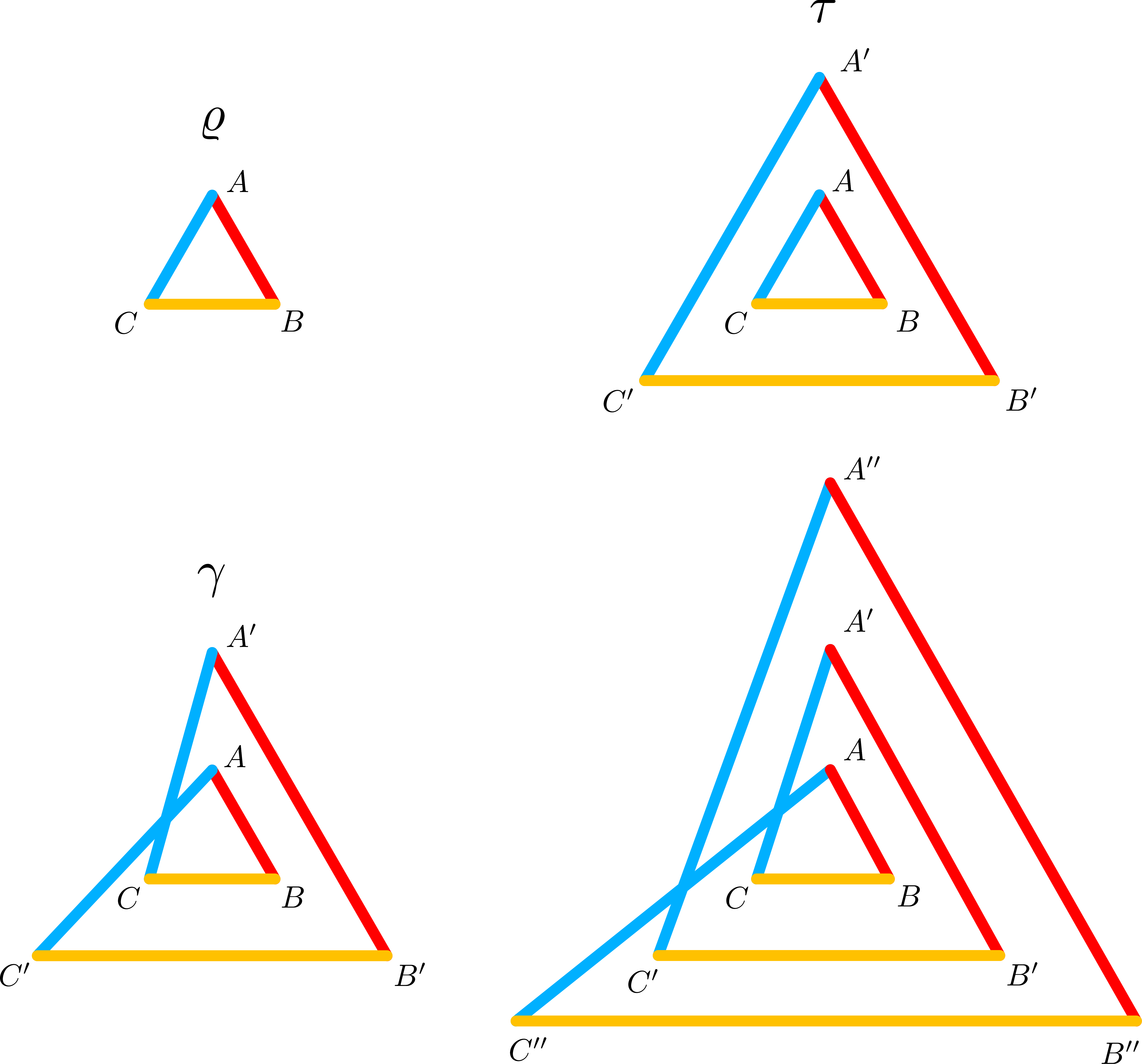}
	\caption{\textit{Triangle network and three of its inflations.} The first figure represents the triangle network Supplementary Fig.\@ \ref{fig:TriangleNetwork}, with global state $\varrho$. Using the same source states (represented by lines of same colour) and same local channels, one can build the so-called inflated state $\tau$, which is biseparable. The state $\gamma$ is build similarly, but with a different rewirering, leading to an inflated state that is in general not separable and different from $\tau$. One may also go to higher order inflations, e.g.\@ with three copies and some rewireing, as depicted here. This procedure implies several equalities between the marginals of the original state and its inflations.
	}
	\label{fig:triangleandinflationsAPPENDIX}
\end{figure}
%%%%%%%%%%%%%%%%%%%%%%%%%%%%%%%%%%%%%%%%%%%%%%%%%%%%%%%%%%%%%%%%%%%%%%%

We start with constructing two inflations of the triangle network. Consider two networks 
with six vertices and six edges as in Supplementary Fig.~\ref{fig:triangleandinflationsAPPENDIX}. 
Identical sources are distributed along the lines of same colour, thus two copies of 
each source are needed per network. In other words, the source $\varrho_b$ is distributed 
between $AC$ and $A'C'$ to generated $\tau$, and between $AC'$ and $A'C$ for $\gamma$ 
(analogously for $\varrho_a$ and $\varrho_c$, following
Supplementary Fig.~\ref{fig:triangleandinflationsAPPENDIX}). Then, the channels are performed 
according to the random parameter $\lambda$. Both on primed and non-primed $A$ 
nodes, the same channel $\mathcal{E}_A^{(\lambda)}$ is applied and similarly for 
$B$ and $C$. This leaves us with two network states, $\tau$ and $\gamma$. Those 
operators are physical states, i.e.\@ they have a unit trace and are positive 
semi-definite. Formally, they can be written as
\begin{align}
 \tau = & \sum_\lambda p_\lambda 
 \Big(\mathcal{E}_A^{(\lambda)} \otimes \mathcal{E}_B^{(\lambda)} \otimes \mathcal{E}_C^{(\lambda)} \big[\varrho_{ABC}\big]\Big) 
 %\nonumber \\
 %&
 \otimes \Big(\mathcal{E}_{A'}^{(\lambda)} \otimes \mathcal{E}_{B'}^{(\lambda)} \otimes \mathcal{E}_{C'}^{(\lambda)}  \big[\varrho_{A'B'C'}\big]\Big)
 \label{eq:gammaI1}
\end{align}
and
\begin{align}
\gamma =  &\sum_\lambda p_\lambda \mathcal{E}_A^{(\lambda)} \otimes \mathcal{E}_B^{(\lambda)} \otimes \mathcal{E}_C^{(\lambda)} \otimes \mathcal{E}_{A'}^{(\lambda)} \otimes 
%\nonumber \\
%&\otimes 
\mathcal{E}_{B'}^{(\lambda)} \otimes \mathcal{E}_{C'}^{(\lambda)}  \big[ \varrho_{ABCA'B'C'}\big],
\end{align}
where $\varrho_{ABCA'B'C'}= \varrho_c \otimes  \varrho_b \otimes  \varrho_a \otimes  \varrho_c \otimes \varrho_b \otimes  \varrho_a$, with the ordering of parties being different on both sides. 
Here, one needs to carefully pay attention to which channel acts on which party (this is 
depicted in Supplementary Fig.\@ \ref{fig:triangleandinflationsAPPENDIX}). Clearly, given only the knowledge
of $\varrho$, the precise form of $\tau$ and $\gamma$ is not known. Still, due to the way 
they are constructed some of their marginals have to be equal, namely
\begin{align}
    \tau_{ABC} &= \tau_{A'B'C'} = \varrho, \\
	\gamma_{ABC} &= \gamma_{A'B'C'}, \\
	\tr_{XX'}(\tau) &= \tr_{XX'} (\gamma) \quad {\rm for} \quad X = A,B,C.
\end{align}
Furthermore, from Eq.\@ (\ref{eq:gammaI1}) it is clear that $\tau$ is separable wrt 
the partition $ABC|A'B'C'$ and we note that $\tau$ and $\gamma$ are permutationally 
symmetric under the exchange of non primed and primed vertices.	Therefore, if, 
for some given state $\varrho$, it is not possible to find states $\tau$ and 
$\gamma$ that satisfy those conditions, then $\varrho$ cannot be generated 
in the considered network.

An interesting point is that the question for the existence of $\tau$ and $\gamma$ with the desired properties can be directly formulated as a semidefinite program (SDP). This can be used to prove that such inflations do not exist, and the corresponding dual program can
deliver an witness-like construction that can be used to exclude preparability of a state
in the network. Still, these approaches are memory intensive. For instance, as the authors
of Ref.~\cite{Navascues2020} acknowledge, it is difficult to derive tests for tripartite qutrit states in a normal computer.

Finally, let us note that other triangle inflations may be considered, for instance inflations with $3n$ nodes ($n=3,4,\dots$) or simply wired differently than $\tau$ and $\gamma$. As mentioned previously, this technique can also be used for more complicated networks. 

%%%%%%%%%%%%%%%%%%%%%%%%%%%%%%%%%%%%%%%%%%%%%%%%%%%%%%%%%%%%%%%%%%%%%
\section{Supplementary Note 2: Fidelity estimate for the GHZ state}
%%%%%%%%%%%%%%%%%%%%%%%%%%%%%%%%%%%%%%%%%%%%%%%%%%%%%%%%%%%%%%%%%%%%%

Let us compute a bound on the fidelity of triangle network states to the GHZ
state, i.e.\@ compute $F=\max \bra{GHZ}\varrho\ket{GHZ}$, where the maximum 
is taken over all states as in Eq.~(\ref{eq:NteStateAPPENDIX}). For this 
maximization is it sufficient to consider the extremal states, which are of 
the type $\varrho_{ITN}=\mathcal{E}_A \otimes \mathcal{E}_B \otimes \mathcal{E}_C \big[\varrho_{ABC}\big]$,
here $ITN$ stands for the independent triangle network, that is the 
triangle network without shared randomness. 

Therefore, one may use techniques based on covariance matrices~\cite{Kraft2020cm,Aberg2020}, 
which are designed for the ITN. 
The covariance matrix (CM) $\Gamma$ of some random variables $x_1, \dots, x_N$ 
is the matrix with elements $\Gamma_{ij} = \mathrm{cov}(x_i,x_j) = \mean{x_ix_j} - \mean{x_i}\mean{x_j}$, $i=1,\dots,N$. 

We can now explain the general idea of the technique in Ref.\@ \cite{Kraft2020cm}: 
If one computes the CM of the outcomes of $Z$-measurements on each qubit 
of an ITN state, then this matrix has a certain block structure. Checking
this block structure can be done by checking the 
positivity %of the eigenvalues
of the comparison matrix $M(\Gamma)$. The comparison matrix obtained by flipping the signs of the off-diagonal elements of $\Gamma$. The condition then reads: For any quantum state in the ITN the comparison matrix of the CM is positive semi-definite. Hence, a negative eigenvalue in the comparison matrix excludes a state of being preparable in the ITN.

Now, if we apply that to states $\varrho(F) = F \ketbra{GHZ} + (1-F) \tilde{\varrho}$ 
with a fidelity $F$ to the GHZ state, the comparison matrix of the CM reads
%\begin{widetext}
\begin{equation}
		M(\Gamma) =
		\begin{bmatrix}
			1- a^2(1-F)^2 & -\left(F + d(1-F)-ab (1-F)^2\right) & -\left(F + e(1-F)-ac (1-F)^2\right) \\
			-\left(F + d(1-F)-ab (1-F)^2\right) & 1- b^2(1-F)^2 & -\left(F + f(1-F)-bc (1-F)^2\right) \\
			-\left(F + e(1-F)-ac (1-F)^2\right) & -\left(F + f(1-F)-bc (1-F)^2\right)& 1-c^2 (1-F)^2  \\
		\end{bmatrix},
	\end{equation}
%\end{widetext}
where $a=\mean{Z11}_{\tilde{\varrho}}$, $b=\mean{1Z1}_{\tilde{\varrho}}$, $c=\mean{11Z}_{\tilde{\varrho}}$, $d=\mean{ZZ1}_{\tilde{\varrho}}$, $e=\mean{Z1Z}_{\tilde{\varrho}}$ and $f=\mean{1ZZ}_{\tilde{\varrho}}$. 
From the last paragraph, we have that this matrix is positive semi-definite 
for ITN states, thus $\bra{\phi}M(\Gamma)\ket{\phi} \geq 0$, for all vectors 
$\ket{\phi}$, and in particular for $\ket{\phi}= (1,1,1)/{\sqrt{3}}$. 
We notice that $\bra{\phi}M(\Gamma)\ket{\phi}$ is upper bounded by $4-6F+F^2$ and
therefore, $0 \leq \bra{\phi}M(\Gamma)\ket{\phi} \leq 4-6F+F^2$ holds for 
ITN states and we are able to exclude all states $\varrho(F)$ with 
$F > 3-\sqrt{5}\simeq 0.7639$ form the triangle network scenario.

By making use of additional  constraints or other criteria, we can obtain tighter bound. 
For any given three compatible dichotomic measurements $M_1, M_2, M_3$, we have~\cite{Navascues2020}
\begin{equation}
    p(M_1=M_2) \geq p(M_1=M_3) + p(M_2=M_3) -1.
\end{equation}
This implies
\begin{equation}
    \mean{M_1M_2} \ge \mean{M_1M_3} + \mean{M_2M_3} - 1.
\end{equation}
By substituting $M_i$ with $-M_i$, we obtain
\begin{align}
    &\mean{M_1M_2} \ge |\mean{M_1M_3} + \mean{M_2M_3}| - 1,\\
    &\mean{M_1M_2} \le 1 - |\mean{M_1M_3} - \mean{M_2M_3}|.
\end{align}
In our case, we have 
\begin{equation}\label{eq:extracon}
    d\ge |a+b| - 1, \quad e\ge |a+c| - 1, \quad f\ge |b+c| - 1.
\end{equation}
With this extra constraint, $0 \leq \bra{\phi}M(\Gamma)\ket{\phi} \leq 9-12F$ holds for 
ITN states and we are able to exclude all states $\varrho(F)$ with 
$F > 3/4 = 0.75$ form the triangle network scenario.

Another criterion for ITN states~\cite{Gisin2020} states that
\begin{align}\label{eq:gisinineq}
      &(1+|E_A|+|E_B|+E_{AB})^2\nonumber\\
    + &(1+|E_A|+|E_C|+E_{AC})^2\nonumber\\ 
    + &(1+|E_B|+|E_C|+E_{BC})^2\nonumber\\
    \le &6(1+|E_A|)(1+|E_B|)(1+|E_C|),
\end{align}
where 
\begin{align}
    &E_{A} = \mean{Z11}_{\varrho}, \ E_{B} = \mean{1Z1}_{\varrho}, \ E_{C} = \mean{11Z}_{\varrho}, \\
    &E_{AB} = \mean{ZZ1}_{\varrho}, \ E_{AC} = \mean{Z1Z}_{\varrho}, \ E_{BC} = \mean{1ZZ}_{\varrho}.
\end{align}

As it turns out, if Eq.~\eqref{eq:gisinineq} together with Eq.~\eqref{eq:extracon} has a feasible solution of $a, b, c, d, e, f \in [-1,1]$, then $F$ should be no more than $1/\sqrt{2} \simeq 0.7071$. Hence, we can exclude all states $\varrho(F)$ with 
$F > 1/\sqrt{2} \simeq 0.7071$ form the triangle network scenario.

In particular, if we know $E_A = E_B =E_C = 0$, i.e., $a=b=c=0$,
Eq.~\eqref{eq:gisinineq} reduces to
\begin{align}
    6 \ge  & \ (1+F+d(1-F))^2 + (1+F+e(1-F))^2\nonumber\\ 
    & + (1+F+f(1-F))^2\nonumber\\
    \ge & \ 3(1+F-(1-F))^2\nonumber\\
    =& \ 12F^2,
\end{align}
which implies $F\le 1/\sqrt{2}$.
The bound $1/\sqrt{2} \simeq 0.7071$ is slightly worse, but close to the one $0.6803$ obtained in Ref.~\cite{Navascues2020} based on advanced numerical computations.
%%%%%%%%%%%%%%%%%%%%%%%%%%%%%%%%%%%%%%%%%%%%%%%%%%%%%%%%%%%%
\section{Supplementary Note 3: Graph and cluster states}
%%%%%%%%%%%%%%%%%%%%%%%%%%%%%%%%%%%%%%%%%%%%%%%%%%%%%%%%%%%
In this Supplementary Note we present our results on graph states and cluster
states. It is structured as follows: We first recall the
basic facts about graph and cluster states. Then, we prove the estimate on the 
fidelity with cluster states for states in the square network (see 
the main text). Finally, we present the proof and discussion of
Observation 1.
 
%%%%%%%%%%%%%%%%%%%%%%%%%%%%%%%%%%%%%%%%%%%%%%%%%%%%%%%%
\subsection{Graph states and the stabilizer formalism}
%%%%%%%%%%%%%%%%%%%%%%%%%%%%%%%%%%%%%%%%%%%%%%%%%%%%%%%%
Graph states~\cite{Hein2004,Hein2006} are quantum states defined 
through a graph $G = (V,E)$, i.e.\@ through a set $V$ of $N$ vertices 
and a set $E$ containing edges that connect the vertices. The vertices 
represent the physical systems, qubits. One way of describing these 
states is through the stabilizer formalism. For that, as introduced 
in the main text, one first needs to introduce the generator operators 
$g_i$ of graph states: a graph state $\ket{G}$ is the unique common 
$+1$-eigenstate of the set of operators $\{g_i\}$,
\begin{equation}
	g_i = X_i \prod_{j \in \mathcal{N}_i} Z_j,
\end{equation}
where $\mathcal{N}_i$ is the neighbourhood of the qubit $i$, i.e.\@ 
the set of all qubits $j\in V$ connected to the qubit $i\in G$. The 
state $\ket{G}$ can also be described through its stabilizer, which is 
the set $\mathcal{S} = \{S_1, \dots, S_{2^N}\} = \{\prod_{i=1}^N g_i^{x_i}:\{x_1,\dots, x_N\}\in\{0,1\}^N\}$. This means that $\mathcal{S}$ contains all possible 
products of the generators $g_i$, hence $S_i \ket{G} = \ket{G}$. We note 
that $\id \in \mathcal{S}$. The projector onto the state $\ket{G}$ reads
\begin{equation} \label{eq:GraphStatesDec}
\ketbra{G} = \frac{1}{2^N} \sum_{i=1}^{2^N} S_i.
\end{equation}

Defined like that, graph states are a subset of the more general stabilizer states~\cite{Gottesman,Audenaert2005,Hein2006}.
First, one has to consider an abelian subgroup $\mathcal{S}$ of the Pauli group $\mathcal{P}_N$ on $N$ qubits that does not contain the operator $-\id$. To that 
set corresponds a vector space $V_\mathcal{S}$ that is said to be stabilized by $\mathcal{S}$, i.e.\@ every element of this vector space is stable under the action of any element of $\mathcal{S}$. We call stabilizers that lead only to one state full-rank stabilizers, i.e.\@ there is a unique common eigenstate with eigenvalue $+1$. That state is completely determined by a subset of $N$ elements of $\mathcal{S}$. As an example, one may consider the GHZ state, as explained in the main text. Indeed, it is the unique common eigenstate of $XXX$, $1ZZ$ and $ZZ1$.  One can show that any stabilizer state is, after a suitable local unitary transformation, equivalent to a graph state.

More precisely, the local unitary transformations that map any stabilizer state to a graph state belong to the so-called \emph{local Clifford group} $\mathcal{C}_1$. The local Clifford group is defined as the normalizer of the single-qubit Pauli group, i.e.\@ $U\mathcal{P}_1 U^{\dagger}=\mathcal{P}_1$ for all $U\in\mathcal{C}_1$. By construction, the stabilizer formalism is preserved under the action of the local Clifford group, and hence, an interesting question is under which conditions two graph states (or two stabilizer states) are equivalent under local Cliffords. For graph states this question has a simple solution in terms of graphical operations that determine their equivalence. Namely, two graph states are equivalent under the action of the local Clifford group if and only if their corresponding graphs are equivalent under a sequence of \emph{local complementations}~\cite{Nest2004}. For a given graph $G$ and vertex $i\in V$ the local complement $G'$ of $G$ at the vertex $i$ is constructed in two steps. First, we have to determine the neighborhood $N(i)\subset V$ of the vertex $i$ and then the induced subgraph is inverted, i.e. considering all possible edges in the neighborhood any pre-existing edge is removed and any non-existing edge is added. E.g. having a graph $V=\{1,2,3\}$ and $E=\{(1,2),(2,3),(3,1)\}$, a local complementation on vertex $1$ results in the graph with edges $E=\{(1,2),(3,1)\}$.

\begin{figure}[t!]
    \centering
    \includegraphics[scale=.15]{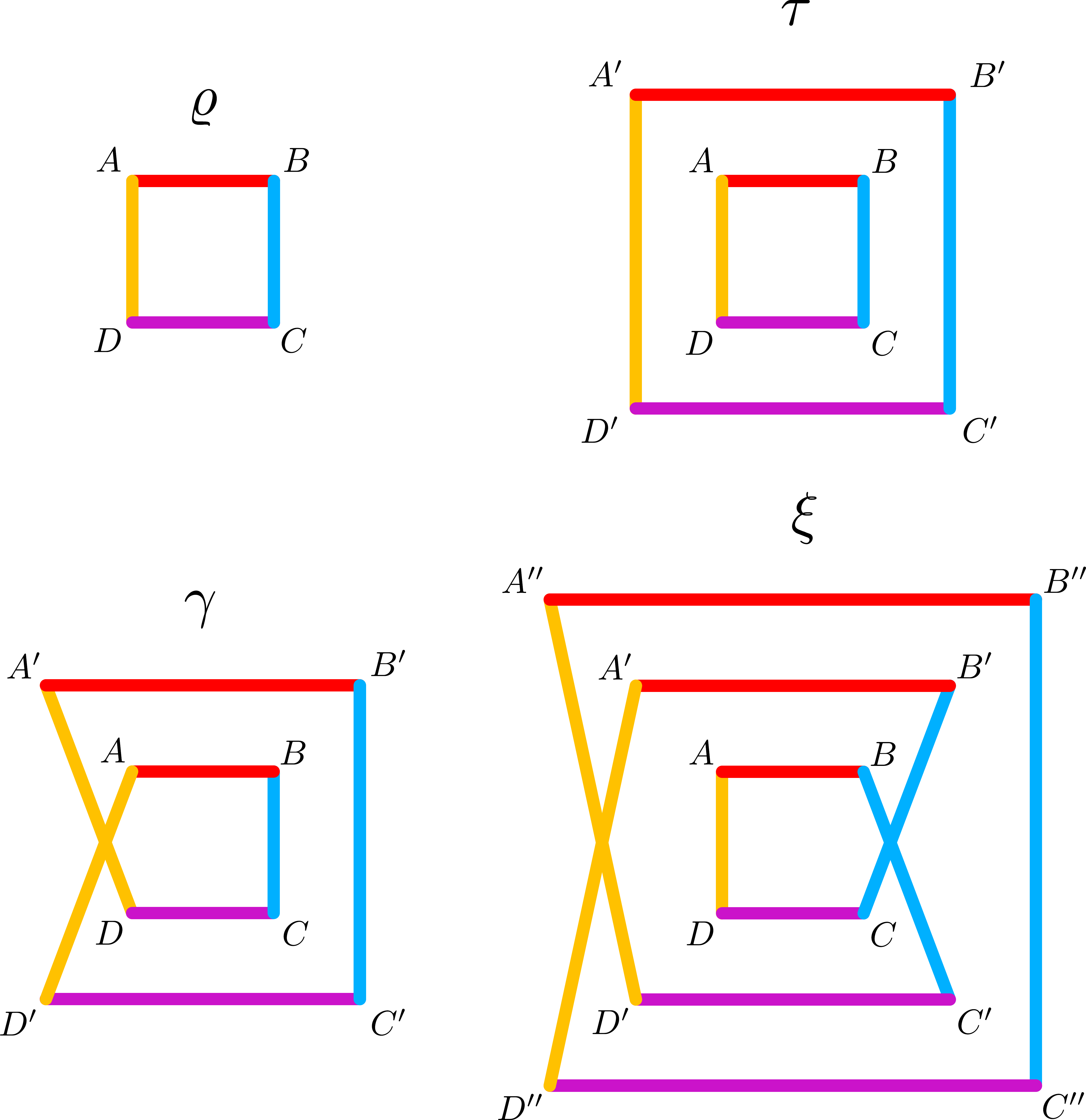}
    \caption{\textit{Square network and three of its inflations.} Similar to the triangle network of Supplementary Fig.\@ \ref{fig:triangleandinflationsAPPENDIX}, the states $\tau$ and $\gamma$ are generated using two copies of the sources and channels used to generate $\varrho$, without and with rewireing respectively. Then, one goes to a high order inflation by using three copies of the sources and the local channels. By rewireing according to the figure, one gets the inflated state $\xi$. Due to the way network states are generated, several of the marginals of $\varrho$, $\tau$, $\gamma$ and $\xi$ are identical.
    }
    \label{fig:Square+I123}
\end{figure}

%%%%%%%%%%%%%%%%%%%%%%%%%%%%%%%%%%%%%%%%%%%%%%%%%%%%%%%%%
\subsection{Estimate for the cluster state fidelity}
%%%%%%%%%%%%%%%%%%%%%%%%%%%%%%%%%%%%%%%%%%%%%%%%%%%%%%%%%
We aim at computing a bound on the fidelity of square network states 
to the four-qubit ring cluster state $\ket{C_4}$, i.e.\@ $F=\mathrm{max}\bra{Cl_4}\varrho\ket{Cl_4}$, were the maximum is 
taken over all square network states $\varrho$. The fidelity of a 
state $\varrho$ with the cluster state is given by $F =\frac{1}{16}\sum_{i=0}^{15}\mean{S_i}_\varrho$, where $\{S_i\}$ is the 
stabilizer of $\ket{C_4}$, consisting of $S_0 = 1111$ and further 
elements given in Supplementary Table \ref{stabilizertable}.

\begin{table*} 
%\begin{center}
    \begin{tabular}{ |c|c|c|c|c| } 
        \hline
        \multicolumn{5}{|c|}{Stabilizer elements} \\ \hline
        $S_1 = XZ1Z$ & $S_5 = YYZZ$ & $S_9 = X1X1$ & $S_{11} = -YXY1$ & $S_{15} = XXXX$\\ 
        $S_2 = ZXZ1$ & $S_6 = YZZY$ & $S_{10}=1X1X$ & $S_{12} = -1YXY$ & \\ 
        $S_3 = 1ZXZ$ & $S_7 = ZYYZ$ &  & $S_{13} = -Y1YX$ & \\ 
        $S_4 = Z1ZX$ & $S_8 = ZZYY$ &  & $S_{14} = -XY1Y$ & \\ \hline
        $\mean{\cdot}_\varrho = \Theta$ & $\mean{\cdot}_\varrho = \Lambda$ & $\mean{\cdot}_\varrho = \Xi$ & $\mean{\cdot}_\varrho = -\Sigma$ & $\mean{\cdot}_\varrho = \Omega$  \\
        \hline
    \end{tabular}
\caption{Elements of the stabilizer of the four-qubit cluster states. The qubit indices $A,B,C,D$ are suppressed here. See the text for further details.}
\label{stabilizertable}
%\end{center}
\end{table*}

The symmetry of $\ket{C_4}$ implies that one can assume the network state $\varrho$ that maximizes the fidelity to admit the same expectation value on operators from the same column, as denoted in the last row of Supplementary Table \ref{stabilizertable}.

As explained in the main text, the general idea is to notice that some stabilizers of $\ket{C_4}$ anticommute in the appropriate inflation, and then use the fact that anticommuting operators cannot all have large expectation values for a given state.
In the $\tau$-inflation of the square network (see Supplementary Fig.\@ \ref{fig:Square+I123}), 
the observable $X_{B}X_{D'}$ and $Y_AY_BZ_CZ_D$ anticommute, and since $\tau_{BD'}=\varrho_{BD}$ and $\tau_{ABCD}=\varrho_{ABCD}$ one has
\begin{equation}
        \Xi^2 + \Lambda^2 \leq 1.
\end{equation}
Secondly, we have Eq.\@ (9) of the main text that we reformulate as
\begin{equation}
    \Xi^2 + \Theta^2 + \Sigma^2 \leq 1.
\end{equation}

At last, we consider the observables $X_AX_BX_CX_D$ and $Z_AX_BZ_{A'}X_{D'}$ in the inflation $\tau$. However, the latter is not a stabilizer 
of the four-qubit ring cluster state, but we have $\mean{Z_AX_BZ_{A'}X_{D'}}_\tau = \mean{Z_AX_BZ_{A'}X_D}_\gamma$. Then, using the fact that for commuting dichotomic measurements, $\mean{M_1 M_2} \geq \mean{M_1 M_3} + \mean{M_2 M_3} - 1$~\cite{Navascues2020}, one gets $\mean{Z_AX_BX_DZ_{A'}}_\gamma \geq \mean{Z_AX_BZ_C}_\gamma +\mean{Z_CX_DZ_{A'}}_\gamma -1$.
Since $X_AX_BX_CX_D$ and $Z_AX_BZ_{A'}X_{D'}$ are anticommuting, from constraints on the marginals, one finally gets
\begin{equation}
    \begin{split}
%        1 &\geq \mean{XXXX}_\tau^2 + \mean{Z_AX_BX_DZ_{A'}}_\gamma^2 \\
%        & \geq \mean{XXXX}_\varrho^2 +  (\mean{ZXZ1}_\varrho +\mean{Z1ZX}_\varrho -1)^2 \\
%        \iff& 
        2\Theta-1 \leq \sqrt{1-\Lambda^2}.
    \end{split}
\end{equation}
Analogously,
\begin{align}
    & 2\Sigma-1 \leq \sqrt{1-\Lambda^2},\\
    & 2\Theta-1 \leq \sqrt{1-\Omega^2},\\
    & 2\Sigma-1 \leq \sqrt{1-\Omega^2}.
\end{align}
By exploiting all these inequalities as constraints on the maximization of the
fidelity, we finally get
\begin{equation}
    \begin{split}
        F &= \frac{1}{16} \big( 1 + 4\Theta + 4 \Lambda + 2 \Xi - 4 \Sigma + \Omega \big)
       % \\ &\leq \frac{21 +10\sqrt{2}}{48} \simeq 0.7321.
        \\ &\leq 0.737684,
    \end{split}
\end{equation}
hence all states with a larger fidelity to the four-qubit ring cluster state cannot be prepared in a square network.

%%%%%%%%%%%%%%%%%%%%%%%%%%%%%%%%%%%%%%%%%%%%%%%%
\subsection{Proof of Observation 1}
%%%%%%%%%%%%%%%%%%%%%%%%%%%%%%%%%%%%%%%%%%%%%%%%

Here we provide a detailed proof of Observation 1. To do so, we first 
need to prove the following theorem.

%%%%%%%%%%%%%%%%%%%%%%%%%%%%%%%%%%%%%%%%%%%%%%%%%%%%
	\begin{figure} 
    \centering
    \includegraphics[scale = .3]{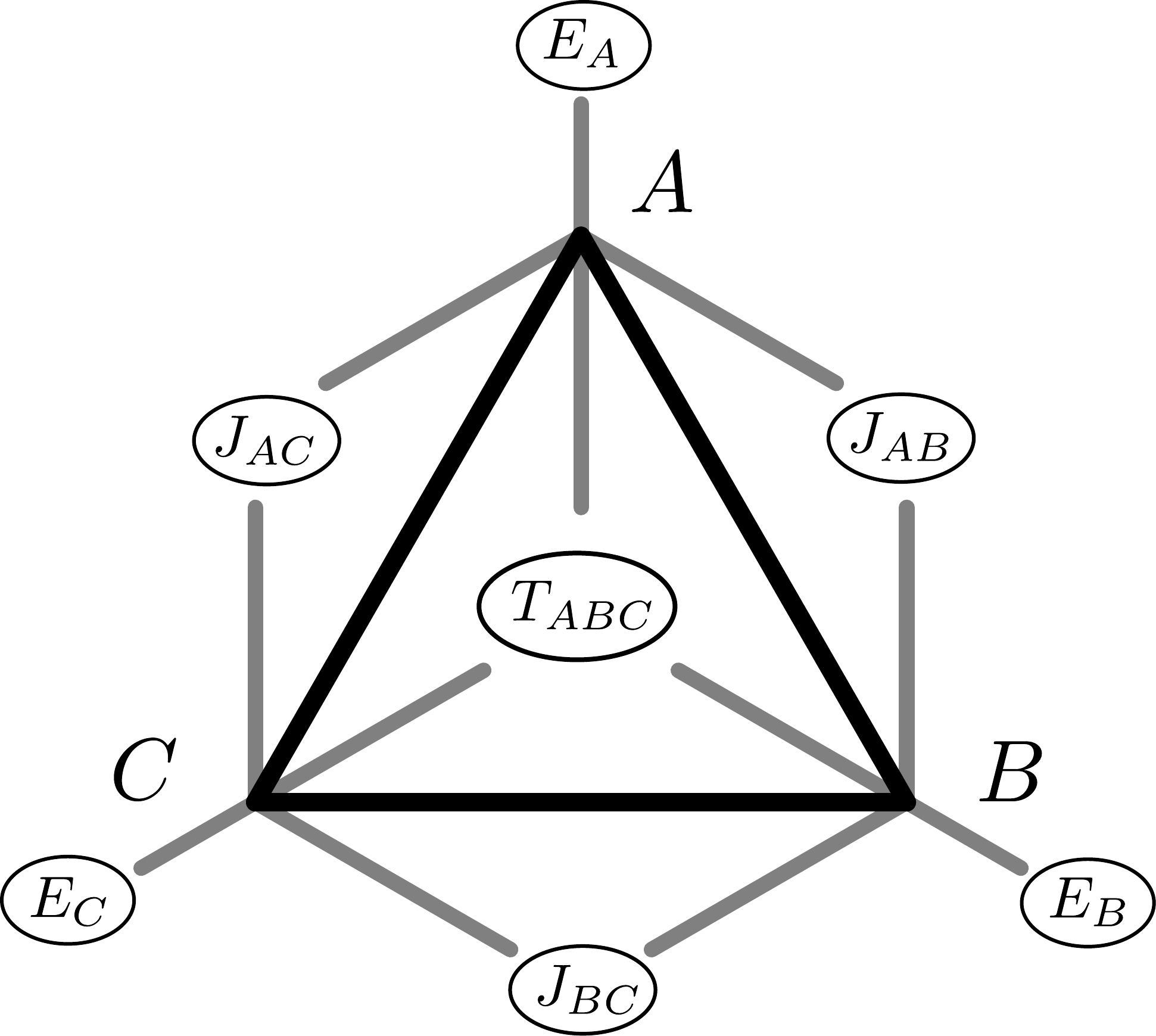}
    \caption{\textit{Illustration of the conditions in Theorem 3.} We consider a triangle
    in the graph of a graph state. The vertices $A,B,C$ share some neighbourhoods, 
    which are indicated by black ellipses. Note that the graph may contain further
    vertices, also the vertices in different neighbourhoods may be connected.
    See the text for further details.}
    \label{fig:ghg}
    \end{figure}
    %%%%%%%%%%%%%%%%%%%%%%%%%%%%%%%%%%%%%%%%%%%%%%%%%%%%
    
\setcounter{theorem}{2}
\begin{theorem}\label{thm:triangle}
	Let $G(V,E)$ be a graph as in Supplementary Fig.~\ref{fig:ghg} with three mutually connected vertices $A$, $B$ and $C$ and let 
	\begin{align}
		& T_{ABC} = \mathcal{N}_A\cap \mathcal{N}_B \cap \mathcal{N}_C,\\ 
		& J_{AB} = (\mathcal{N}_A\cap \mathcal{N}_B)\setminus T_{ABC},\\
		& E_A = \mathcal{N}_A \setminus (\mathcal{N}_B\cup \mathcal{N}_C),
	\end{align}
	etc., where $\mathcal{N}_X$ is the neighborhood of $X$ ($X=A,B,C$).
	Then the graph state $\ket{G}$ cannot originate from any network with only
	bipartite sources, if one of the following conditions is satisfied:
	\begin{enumerate}
		\item $J_{XY} = J_{XZ} = \emptyset$, where $X,Y,Z$ is a permutation of $A, B, C$;
		\item $E_X = E_Y = \emptyset$, where $X\neq Y \in \{A, B, C\} $ ;
		\item $E_X = J_{XY} = \emptyset$, where  $X\neq Y \in \{A, B, C\} $.
	\end{enumerate}
\end{theorem}
\begin{proof}
	We only need to show that the graph state $\ket{G}$ cannot be generated in the network with the complete graph $K$ as shown in Supplementary Fig.~\ref{fig:ghk}, where the number of vertices is the same than in $G$. 
	%%%%%%%%%%%%%%%%%%%%%%%%%%%%%%%%%%%%%%%%%%%%%%%%%%%
	\begin{figure} 
    \centering
    \includegraphics[scale = .2]{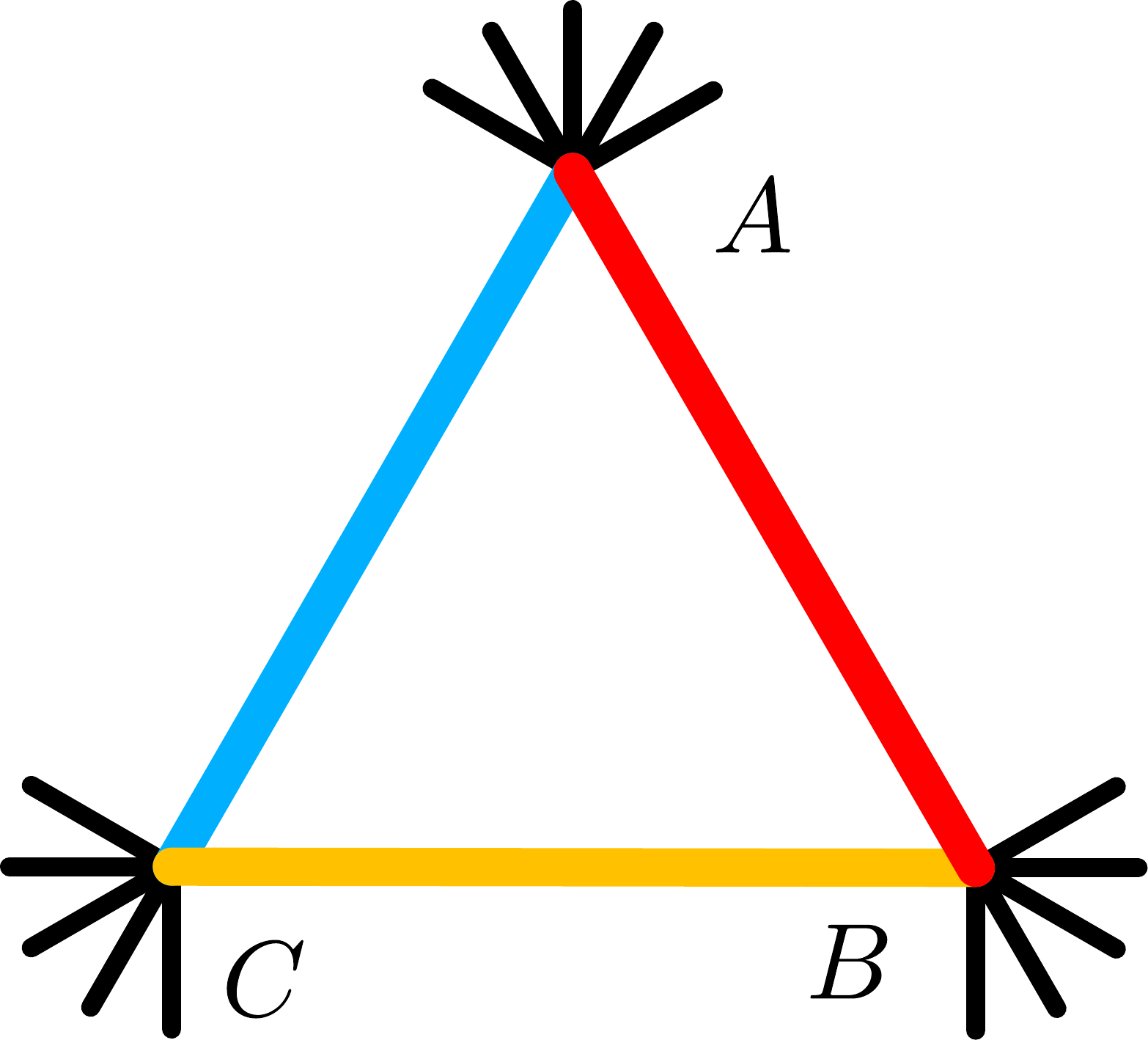}
    \caption{\textit{The bipartite network with complete graph $K$.} This is the network that should generate the graph state from Supplementary Fig.~\ref{fig:ghg}. Note that here the bipartite links of the network are shown, and not the edges of the graph of the graph state.}
    \label{fig:ghk}
    \end{figure}
    %%%%%%%%%%%%%%%%%%%%%%%%%%%%%%%%%%%%%%%%%%%%%%%%%%%
	%We consider two inflations of a state $\varrho$ from the network $K$, namely %$\gamma$ and $\eta$ as in Fig.~\ref{fig:ghgamma} and Fig.~\ref{fig:gheta}. 
	
	We start our discussion with the inflation $\gamma$ as in Supplementary Fig.~\ref{fig:ghgamma}. 
	%%%%%%%%%%%%%%%%%%%%%%%%%%%%%%%%%%%%%%%%%%%%%%%%%%%%
	\begin{figure} 
    \centering
    \includegraphics[scale = .2]{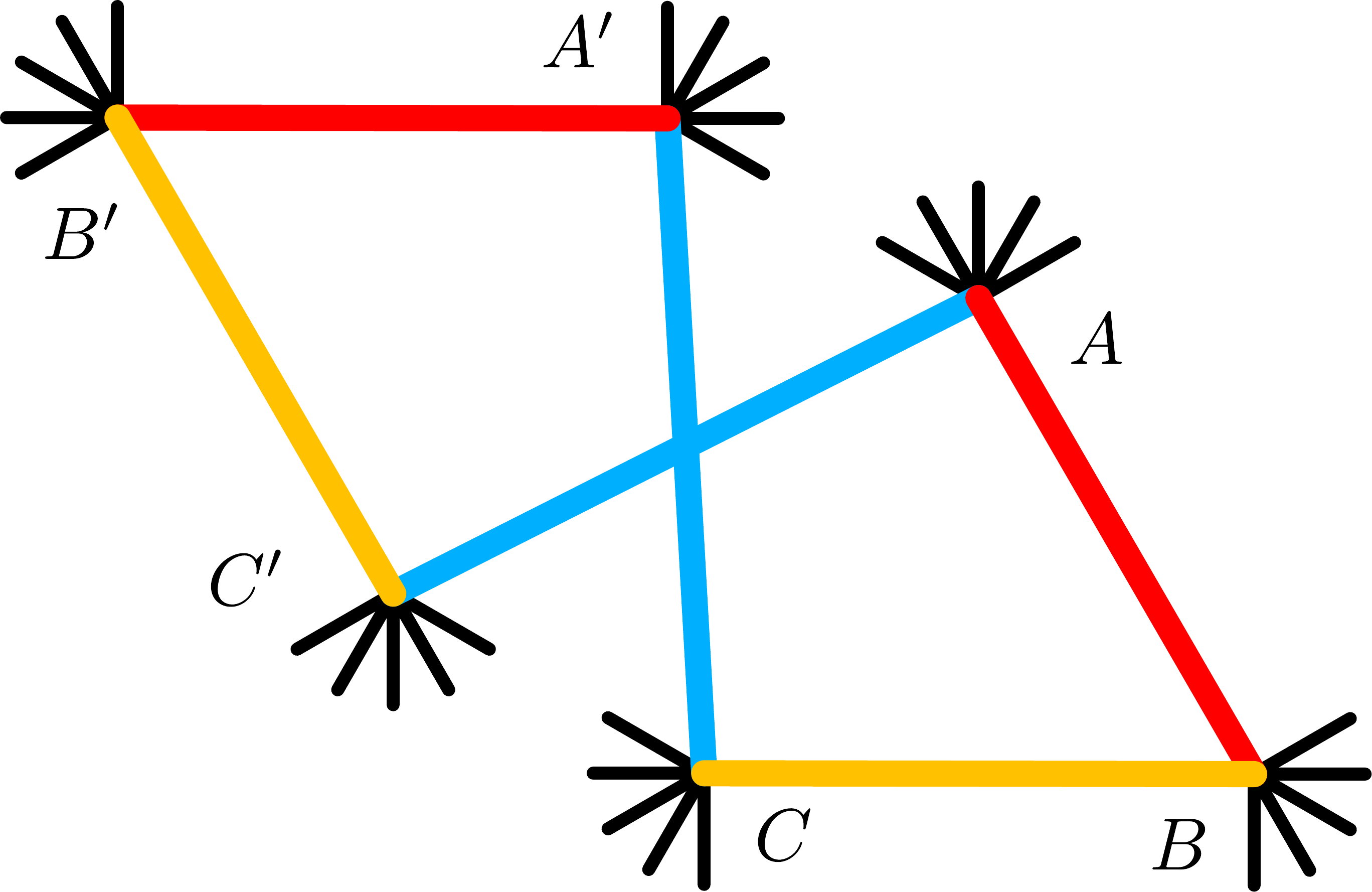}
    \caption{\textit{The network for the inflation $\gamma$. }This is a two-copy 
    inflation, where only the links between $AC$ and $A'C'$ are rewired.}\label{fig:ghgamma}
    \end{figure}
    %%%%%%%%%%%%%%%%%%%%%%%%%%%%%%%%%%%%%%%%%%%%%%%%%%%%%%%%
    Let us denote 
	\begin{equation}
		g_A = X_A Z_{\mathcal{N}_A}, \quad g_B = X_B Z_{\mathcal{N}_B}, \quad g_C = X_C Z_{\mathcal{N}_C},
	\end{equation}
	where $\mathcal{N}_A$ is the neighborhood of the vertex $A$ in the graph $G$.
	Then
	\begin{align}
		& g_A g_B = Y_AY_B Z_{R_{AB}},\\ 
		& g_A g_C = Y_AY_C Z_{R_{AC}},\\ 
		& g_B g_C = Y_BY_C Z_{R_{BC}},
	\end{align}
	where $R_{AB} = {E_A\cup E_B\cup J_{AC}\cup J_{BC}}$, and analoguously
	for $R_{AC}$ and $R_{BC}$.
	Since $g_A g_C = (g_A g_B)(g_B g_C)$, we can apply the usual argument
	from the GHZ state to conclude that
	\begin{equation}\label{eq:cond0}
		\mean{Y_AY_C Z_{R_{AC}}}_\gamma \ge \mean{Y_AY_BZ_{R_{AB}}}_\gamma + \mean{Y_BY_CZ_{R_{BC}}}_\gamma - 1.
	\end{equation}
	By comparing the marginals of the states $\gamma$ and $\varrho$, we have
	\begin{align}\label{eq:cond0m}
		& \mean{Y_AY_BZ_{R_{AB}}}_\gamma = \mean{Y_AY_BZ_{R_{AB}}}_\varrho,\\ 
		& \mean{Y_BY_CZ_{R_{BC}}}_\gamma = \mean{Y_BY_CZ_{R_{BC}}}_\varrho,
	\end{align}
	since $A, C \not\in R_{AB} \cup R_{BC}$. In the following, we will also
	use the notation 
	\begin{equation}
	\mathcal{R} = R_{AC}= E_A\cup E_C\cup J_{AB}\cup J_{BC}
	\end{equation}
	in order to avoid a plethora of indices. 
	
	Then, we consider the inflation $\eta$ shown in Supplementary Fig.~\ref{fig:gheta}.
	This is constructed as follows. First, one has two disconnected complete
	graphs, $K$ and $K'$. Then, one takes the subset $\mathcal{R}$ as a subgraph
	of $K$ and rewires all connections from vertices in $\mathcal{R}$ to $A$ to 
	$A'$. Similarly, one takes the subset $\mathcal{R'}$ as a subgraph
	of $K'$ and rewires all connections from vertices in $\mathcal{R'}$ to $A'$ to 
	$A$.

	By comparing the marginals of $\gamma$ and $\eta$, this time we have
	\begin{equation}\label{eq:gammaeta}
		\mean{Y_AY_CZ_{\RR}}_\gamma = \mean{Y_{A'}Y_{C}Z_{\RR}}_\eta.
	\end{equation}
	With this reasoning, we have established that the correlation $\mean{Y_{A'}Y_{C}Z_{\RR}}_\eta$ is large in $\eta$, if the original
	state $\varrho$ is close to the graph state. Now we have to identify
	another anticommuting observables in $\eta$ with large expectation value
	in order to arrive at a contradiction. 
	
	A natural first candidate is the stabilizing operator
	\begin{equation}
	{g}_B = X_B Z_AZ_C Z_{\mathfrak{R}}.
	\end{equation}
	with $\mathfrak{R} =E_{B}\cup J_{AB} \cup J_{BC} \cup T_{ABC}$
	of the graph state. This, however, cannot always be identified
	with some observable in the inflation $\eta$.
	Still, if 
	\begin{equation}
	\RR \cap \mathfrak{R} = \emptyset \Leftrightarrow J_{AB} = J_{BC} = \emptyset
	\end{equation}
	the observable ${g}_B$ is not affected by the rewiring in $\eta$, we 
	have $\mean{{g}_B}_\varrho=\mean{{g}_B}_\eta$. Moreover, in $\eta$
	the observables $g_B$ and $Y_A'Y_CZ_{\RR}$ anticommute. So, in this case we have
	\begin{equation}
	\mean{Y_{A}Y_{C}Z_{\RR}}_\gamma^2 + \mean{{g}_B}_\eta^2 \le 1.
	\end{equation}
	and for the original $\varrho$ we arrive at the condition (assuming $\mean{Y_A Y_B Z_{R_{AB}}}_\varrho + \mean{Y_BY_CZ_{R_{BC}}}_\varrho - 1\geq 0$, as in Eq.~(6)
	in the main text)
	\begin{equation}
	 (\mean{Y_A Y_B Z_{R_{AB}}}_\varrho + \mean{Y_BY_CZ_{R_{BC}}}_\varrho - 1)^2 
	 + \mean{{g}_B}_\varrho^2 \leq 1
	\end{equation}
    for states that can be prepared in the network.
    % remove square
	%%%%%%%%%%%%%%%%%%%%%%%%%%%%%%%%%%%%%%%%%%%%%%%%%%%%%%%%%%%%%%%%%%%%%
    \begin{figure} 
    \centering
    \includegraphics[scale = .20]{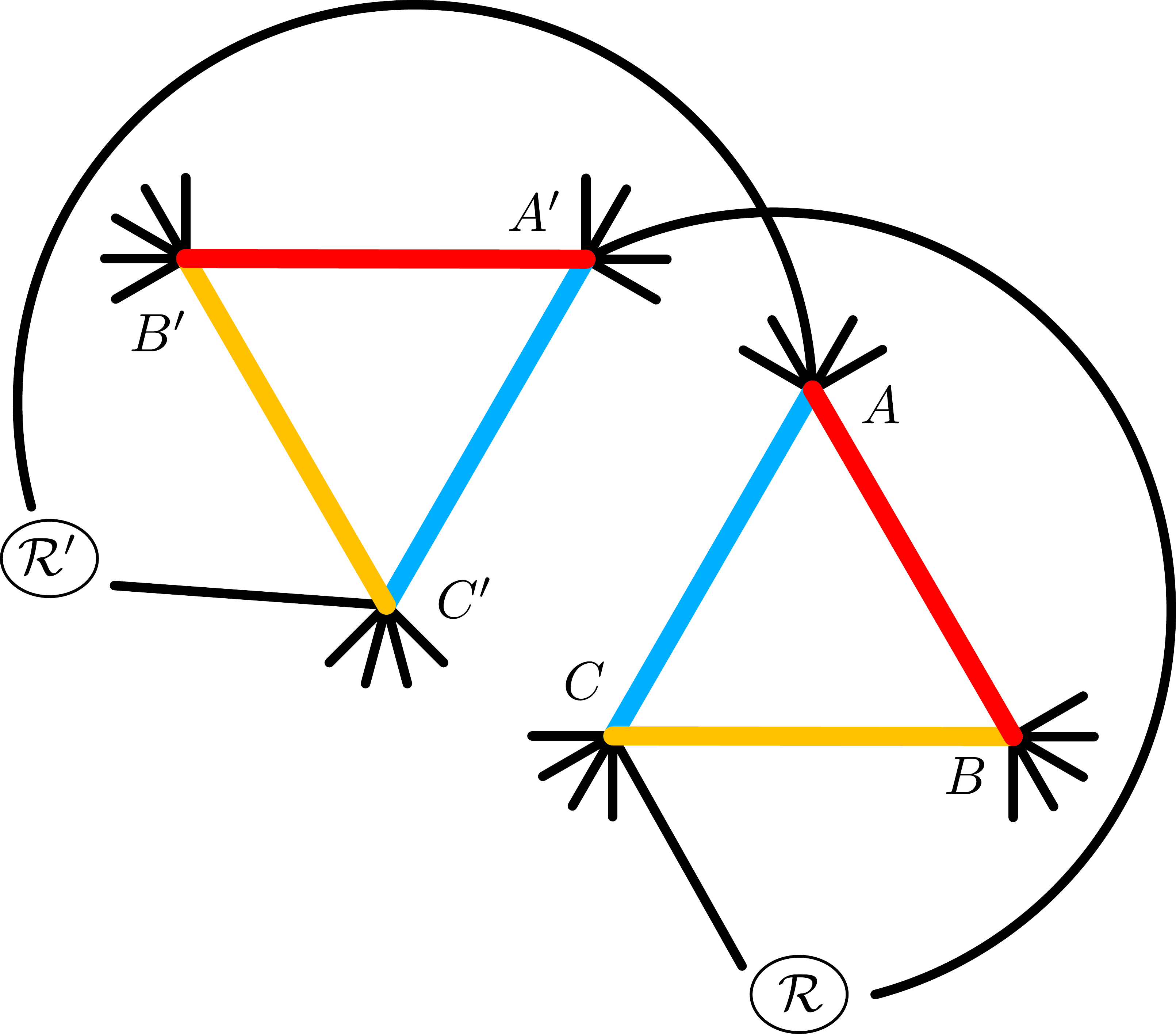}
    \caption{\textit{The network for the inflation $\eta$.} This is constructed as follows. First, one has two disconnected complete
	graphs, $K$ and $K'$. Then, one takes the subset $\mathcal{R}$ as a subgraph
	of $K$ and rewires all connections from vertices in $\mathcal{R}$ to $A$ to 
	$A'$. Similarly, one takes the subset $\mathcal{R'}$ as a subgraph
	of $K'$ and rewires all connections from vertices in $\mathcal{R'}$ to $A'$ to 
	$A$.}
    \label{fig:gheta}
    \end{figure}
	%%%%%%%%%%%%%%%%%%%%%%%%%%%%%%%%%%%%%%%%%%%%%%%%%%%%

	Furthermore, in the case that $E_A = E_C = \emptyset$ and by making 
	use of the operator
	\begin{equation}
	{g}_A{g}_B{g}_C = X_A X_B X_C Z_{E_B \cup T_{ABC}},  
	\end{equation}
	we arrive at a similar condition on $\varrho$ that is also 
	violated by the graph state $\ket{G}$.
	
	Finally, if $E_A = J_{AB} = \emptyset$, one can make use of 
	the operator
	\begin{equation}
	{g}_A = X_AZ_BZ_CZ_{J_{AC} \cup T_{ABC}},  
	\end{equation}
	in order to arrive at a condition on $\varrho$ that is not satisfied by 
	the graph state $\ket{G}$.
	 
	 By permuting $A, B, C$ in the above argument, we finish our proof.
\end{proof}

In the following, we identify some basic situations where the Theorem 
\ref{thm:triangle} can be applied. First, we show that the conditions
of Theorem \ref{thm:triangle} are met, if there is one vertex with a 
small degree.

%%%%%%%%%%%%%%%%%%%%%%%%%%%%%%%%%%%%%%%%%%%%%%%%%%%%%%%%%%%%%%%%%%%%%%%%%%%5
\begin{corollary}\label{coro:d3}
	Let $G$ be connected graph with no less than three vertices. If its 
	minimal degree is no more than three, the graph state $\ket{G}$ cannot 
	be generated by any network with bipartite sources.
\end{corollary}
%%%%%%%%%%%%%%%%%%%%%%%%%%%%%%%%%%%%%%%%%%%%%%%%%%%%%%%%%%%%%%%%%%%%%%%%%%%%%%%%%%%
\begin{proof}
The proof will be done successively for minimal degree one, two and three.
	
	Let $v$ be a vertex whose degree is one and let $w$ be the vertex connected to $v$. Since $G$ is a connected graph with no less than three vertices,  
	\begin{equation}
		\mathcal{N}_w \setminus \{v\} \neq \emptyset, 
	\end{equation}
	where $\mathcal{N}_w$ is the neighbourhood of $w$. If we apply local 
	complementation on the vertex $w$, we obtain a new graph  $G'$, where 
	\begin{equation}
	u \sim v, \quad \forall u\in \mathcal{N}_w \setminus \{v\},
	\end{equation}
	where $u\sim v$ means that the vertices $u, v$ are connected.
	
	By setting
	\begin{equation}
		B = w, \quad A = v, \quad C = u_0,
	\end{equation}
	where $u_0$ is an arbitrary vertex in $ \mathcal{N}_w \setminus \{v\} $, 
	we see that
	\begin{equation}
	    \begin{split}
	        &\mathcal{N}_B \setminus \{A,C\} = T_{ABC} \cup J_{AB}, \\ &A\sim B, \quad A\sim C, \quad B\sim C.
	    \end{split}
	\end{equation}
	Hence, $E_B = J_{BC} = \emptyset$, which implies  $\ket{G'}$ cannot be from any network with only bipartite sources. Since  $\ket{G}$ is equivalent to  $\ket{G'}$ up to a local unitary transformation, we come to the same conclusion for $\ket{G}$.
	
	Now, let us consider graphs with minimal degree equal to two, and let $v$ be a vertex with degree two, and $w$ and $u$ be the two vertices connected to  $v$. If  $w \not\sim u$, we can apply a local complementation on $v$ to connect them. Hence, we can assume  $w \sim u$ without loss of generality. By setting
	\begin{equation}
		A = w, \quad B = v, \quad C = u,
	\end{equation}
	we have
	\begin{equation}
		E_B = J_{AB} = J_{BC} = T_{ABC} = \emptyset,
	\end{equation}
	which leads to the desired conclusion.
	
	Lastly, let $v$ be a vertex with degree three and let $w$, $u$ and $t$ be the three vertices connected to $v$. Since we can apply local complementation on $v$, without loss of generality, we can assume that there are at least two edges among  $w$, $u$ and $t$, more specifically,  $w\sim u$ and $w\sim t$. Let us take
	\begin{equation}
		A = w, \quad B = v,  \quad  C = u,
	\end{equation}
	hence we see that
	\begin{equation}
		E_B = \emptyset, \quad J_{AB} = \emptyset \quad \text{or} \quad J_{BC}=\emptyset.
	\end{equation}
	This implies that the graph state $\ket{G}$ with a vertex of degree three cannot be from any network with only bipartite sources.
\end{proof}

Under certain conditions, we can also exclude a network structure for
graphs with minimal degree four:

%%%%%%%%%%%%%%%%%%%%%%%%%%%%%%%%%%%%%%%%%%%%%%%%%%%%%%%%%%%%%%%%%%%%%%%%%%
\begin{corollary}\label{coro:d4}
	Let $G$ be a graph that has a vertex $v$ of degree four. If the induced subgraph on the neighborhood $\mathcal{N}_v$ is not a line graph, then $\ket{G}$ cannot be generated in any network with  bipartite sources. 
\end{corollary}
%%%%%%%%%%%%%%%%%%%%%%%%%%%%%%%%%%%%%%%%%%%%%%%%%%%%%%%%%%%%%%%%%%%%%%%%%%%
\begin{proof}
	As shown in Supplementary Fig.~\ref{fig:d4}, there are six inequivalent graphs with four vertices up to permutation and complementation. 
	%%%%%%%%%%%%%%%%%%%%%%%%%%%%%%%%%%%%%%%%%%%%%%%%%%
	\begin{figure} 
    \centering
    \includegraphics[width=0.42\textwidth]{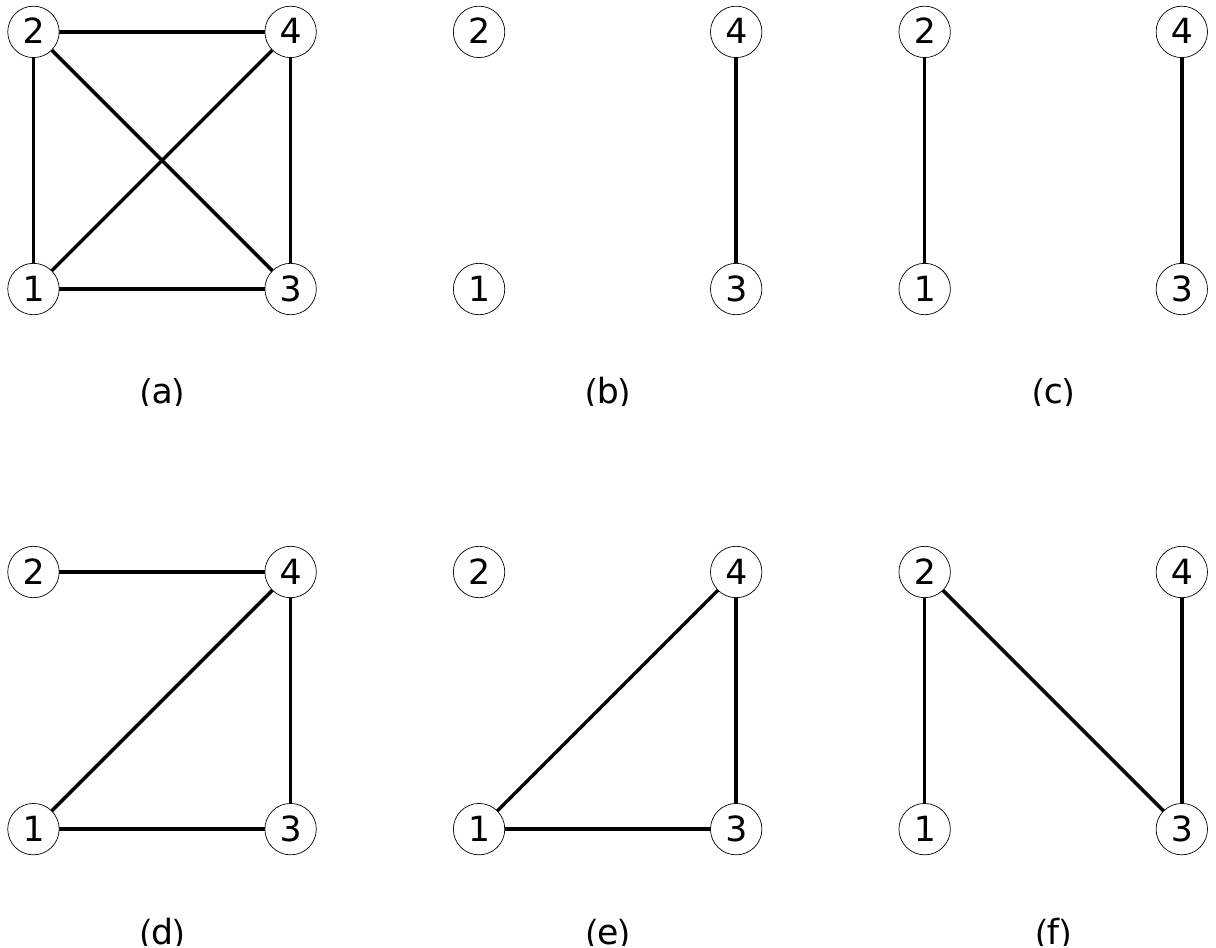}
    \caption{The six inequivalent graphs with four vertices up to permutations and complementation.}
    \label{fig:d4}
    \end{figure}
  %%%%%%%%%%%%%%%%%%%%%%%%%%%%%%%%%%%%%%%%%%%%%%%%%%
	In case (a), we can set $B = v, A = u_1, C = u_2$, then
	\begin{equation}
		T_{ABC} = \{u_3, u_4\}, \quad E_B = J_{AB} = J_{BC} = \emptyset.
	\end{equation}
	In case (b) and (c), we can set $B=v, A = u_3, C=u_4$, then
	\begin{equation}
		E_B = \{u_1, u_2\}, \quad J_{AB} = J_{BC} = T_{ABC} = \emptyset.
	\end{equation}
	In case (d) and (e), we can set $B=v, A = u_1, C = u_3$, then
	\begin{equation}
		E_B = \{u_2\}, \quad T_{ABC} = \{u_4\}, \quad J_{AB} = J_{BC} = \emptyset.
	\end{equation}
    In all the above cases, Theorem~\ref{thm:triangle} implies that the graph state $\ket{G}$ cannot be from any network with only bipartite sources.
	In case (f), the neighbourhood $\mathcal{N}_v$ of $v$ is a line graph whose complementation is also a line graph.
\end{proof}

Having established these results, we can discuss graphs with a small number
of vertices. Here. previous works have established a classification of all 
small graphs with respect to equivalence classes under local complementation. 
In detail, this classification has been achieved for up to seven vertices in Ref.~\cite{Hein2004},
for eight vertices in  Ref.~\cite{Adan2009} and for nine to twelve vertices in Ref.~\cite{Danielsen2011}.
These required numerical techniques are advanced, as, for 
instance, for twelve qubits already $1~274~068$ different equivalence classes under local 
complementation exist. We can use this classification now, and apply our result
on it to obtain: 

%%%%%%%%%%%%%%%%%%%%%%%%%%%%%%%%%%%%%%%%%%%%%%%%%%%%%%%%%%%%%%%%%%%%%%%%%%
\begin{theorem}
No graph state with up to twelve vertices can originate from a network with 
only bipartite sources.
\end{theorem}
%%%%%%%%%%%%%%%%%%%%%%%%%%%%%%%%%%%%%%%%%%%%%%%%%%%%%%%%%%%%%%%%%%%%%%%%%
\begin{proof}
    Using the tables in Ref.~\cite{Danielsen2011} one can directly check that except the graph 
    $G_{d5}$ in Supplementary Fig.~\ref{fig:d5}, all graphs with no more than twelve vertices, 
    up to isomorphism and local complementation, satisfy at least one condition 
    in Corollary~\ref{coro:d3} and~\ref{coro:d4}. 
    
    For the graph $G_{d5}$, the minimal degree is no less than $5$ whatever local complementation is applied. However, if we set
    \begin{equation}
        B = v_1, \quad A = v_4, \quad C = v_5,
    \end{equation}
    then
    \begin{equation}
        E_B = \{v_2, v_3\}, \quad T_{ABC} = \{v_6\}.
    \end{equation}
    Thus, $J_{AB} = J_{BC} = \emptyset$, which implies that the graph state $\ket{G_{d5}}$ cannot originate from a network with only bipartite sources.
    
    All in all, no graph state with more than twelve vertices can originate 
    from a network with only bipartite sources.
\end{proof}

\begin{figure}
    \centering
    \includegraphics[scale = .5]{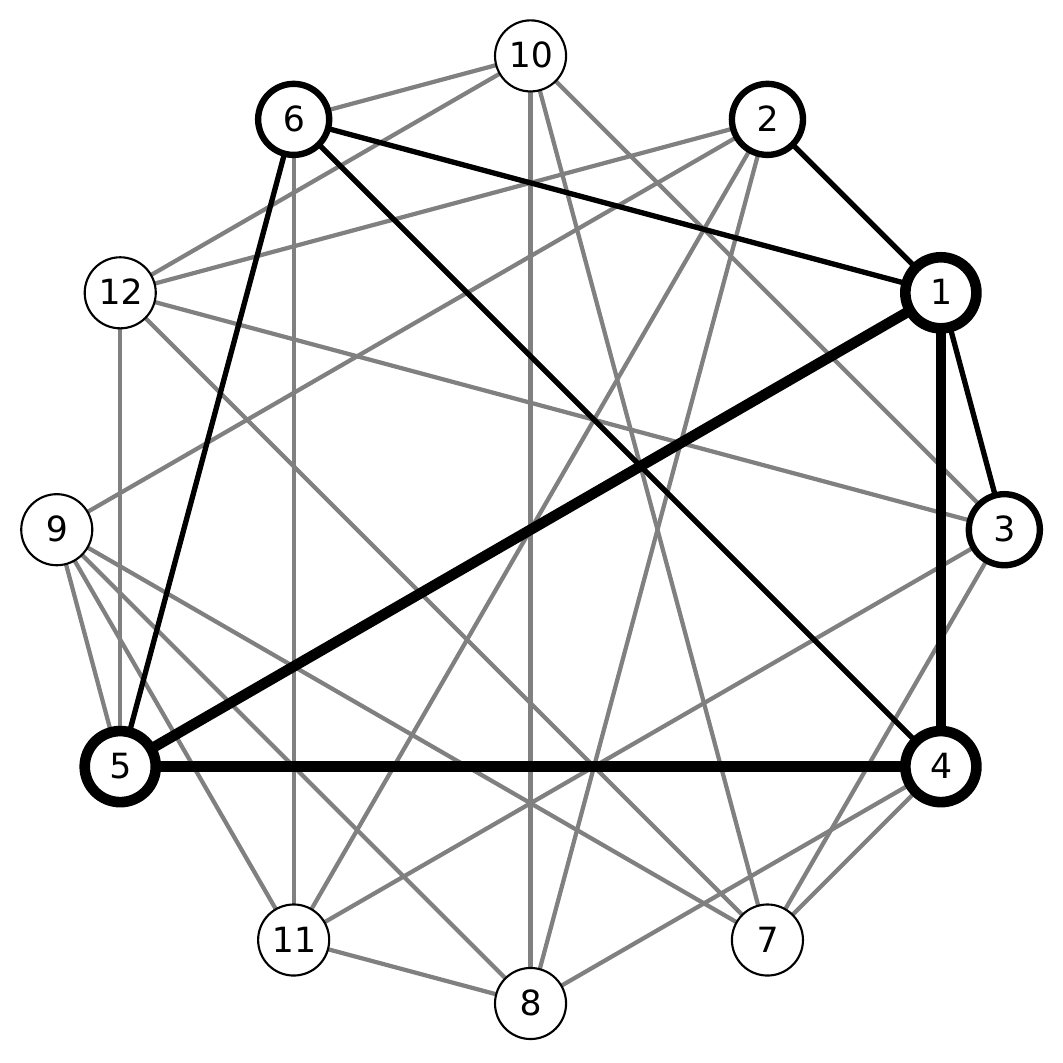}
    \caption{The graph $G_{d5}$ with twelve vertices, where the minimal degree is no less than five whatever local complementation is applied. }\label{fig:d5}
\end{figure}

The statements in Observation 1 concerning the two- or three-dimensional cluster states
follow directly from Corollary 5 (in the 2D case) or the application of a local complementation and Theorem 3.

\subsection{Networks with tripartite sources}

Let us finally discuss two examples of networks with tripartite sources. First, one may expand the method used in the main text to exclude the GHZ state of the set of network states. Consider the fully connected four-partite graph state, which is locally equivalent to the four-partite GHZ state and let us show that it cannot have originated from the four partite network with four tripartite sources, illustrated in Supplementary Fig.\@ \ref{fig:tripartsquare}. Consider the $\nu$- and $\tau$-inflations of that network, depicted in Supplementary Fig.\@ \ref{fig:tripartsquareINFL}. From $\mean{Y_iY_k}_\nu \geq \mean{Y_iY_j}_\nu + \mean{Y_jY_k}_\nu -1$ and the marginal equalities, one gets $\mean{Y_AY_{A'}}_\tau \geq \mean{Y_AY_B}_\varrho + \mean{Y_BY_C}_\varrho+ \mean{Y_CY_D}_\varrho+ \mean{Y_DY_A}_\varrho -3$. Then, from an  anticommuting relation in the $\tau$-inflation and marginal equalities, one gets $\mean{Y_AY_{A'}}_\tau^2+\mean{X_AZ_BZ_CZ_D}_\varrho^2 \leq 1$. Combining those two equations, one gets that a fidelity higher than $0.923$ to the the fully connected graph state cannot be achieved by a four-partite network state.

Second, consider the graph state $\ket{G}$ whose graph is represented in Supplementary Fig.\@ \ref{fig:additional2} (note that it is equivalent under local complementation to the fully connected graph state with the edge between $A$ and $D$ missing). Now, consider a six-partite network with six tripartite sources, connected as in Supplementary Fig.\@ \ref{fig:triparthexa}.  One can show that the observables $X_AX_{D''}$, $Z_AX_BZ_D$ and $Y_AY_{E'}X_D$  anticommute in the $\xi$-inflation of the network (see Supplementary Fig.\@ \ref{fig:triparthexaINLF}) and thus get to an inequality for any state $\varrho$ in the considered network, $\mean{X_AX_{D}}_\varrho^2 + \mean{Z_AX_BZ_D}_\varrho^2 + \mean{Y_AY_{E}X_D}_\varrho^2 \leq 1$, where all three observables are in the stabilizer of $\ket{G}$. Seven other inequalities of that type can be derived, and using them one gets that no state with a fidelity to $\ket{G}$ higher that $0.949$ can be prepared in the quantum network of Supplementary Fig.\@ \ref{fig:triparthexa}.
%%%%%%%%%%%%%%%%%%%%%%%%%%%%%%%%%%%%%%%%%%%%%%%%%%%%
\begin{figure}
    \centering
    \includegraphics[scale=.2]{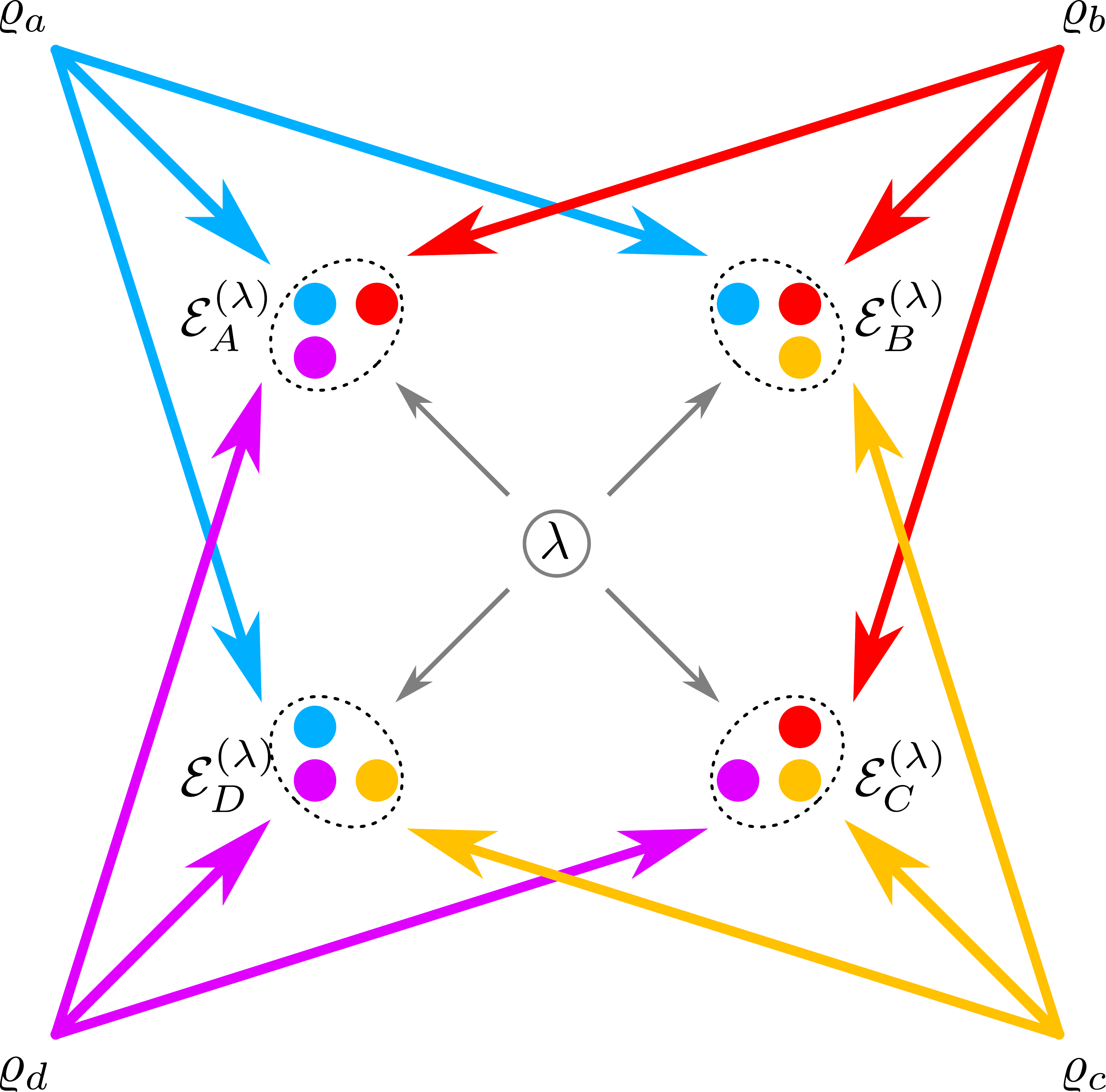}
    \caption{Square network with four tripartite sources. }\label{fig:tripartsquare}
\end{figure}
%%%%%%%%%%%%%%%%%%%%%%%%%%%%%%%%%%%%%%%%%%%%%%%%%%%%    
\begin{figure}
    \centering
    \includegraphics[scale=.1]{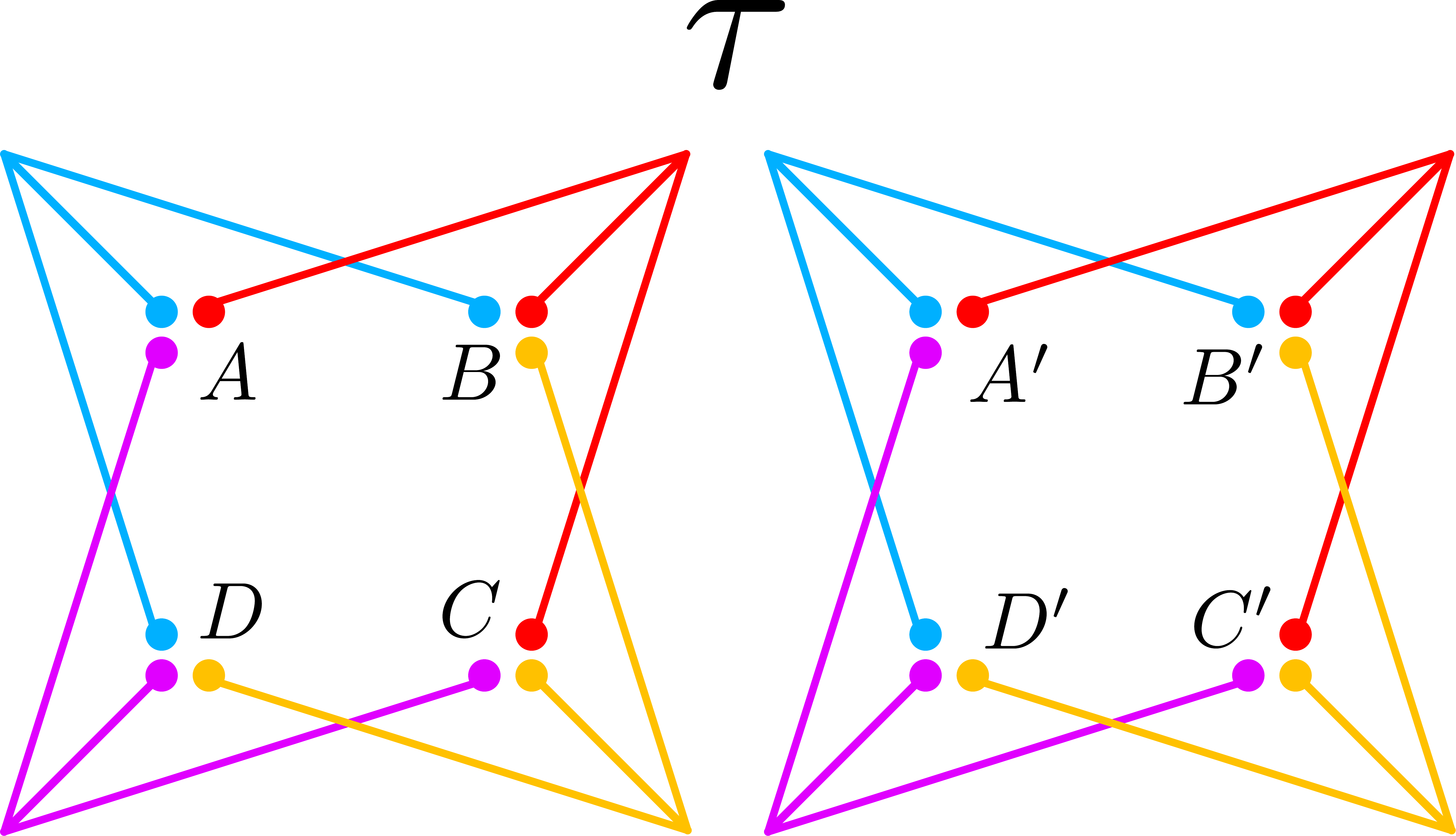} \hspace{1cm}
    \includegraphics[scale=.1]{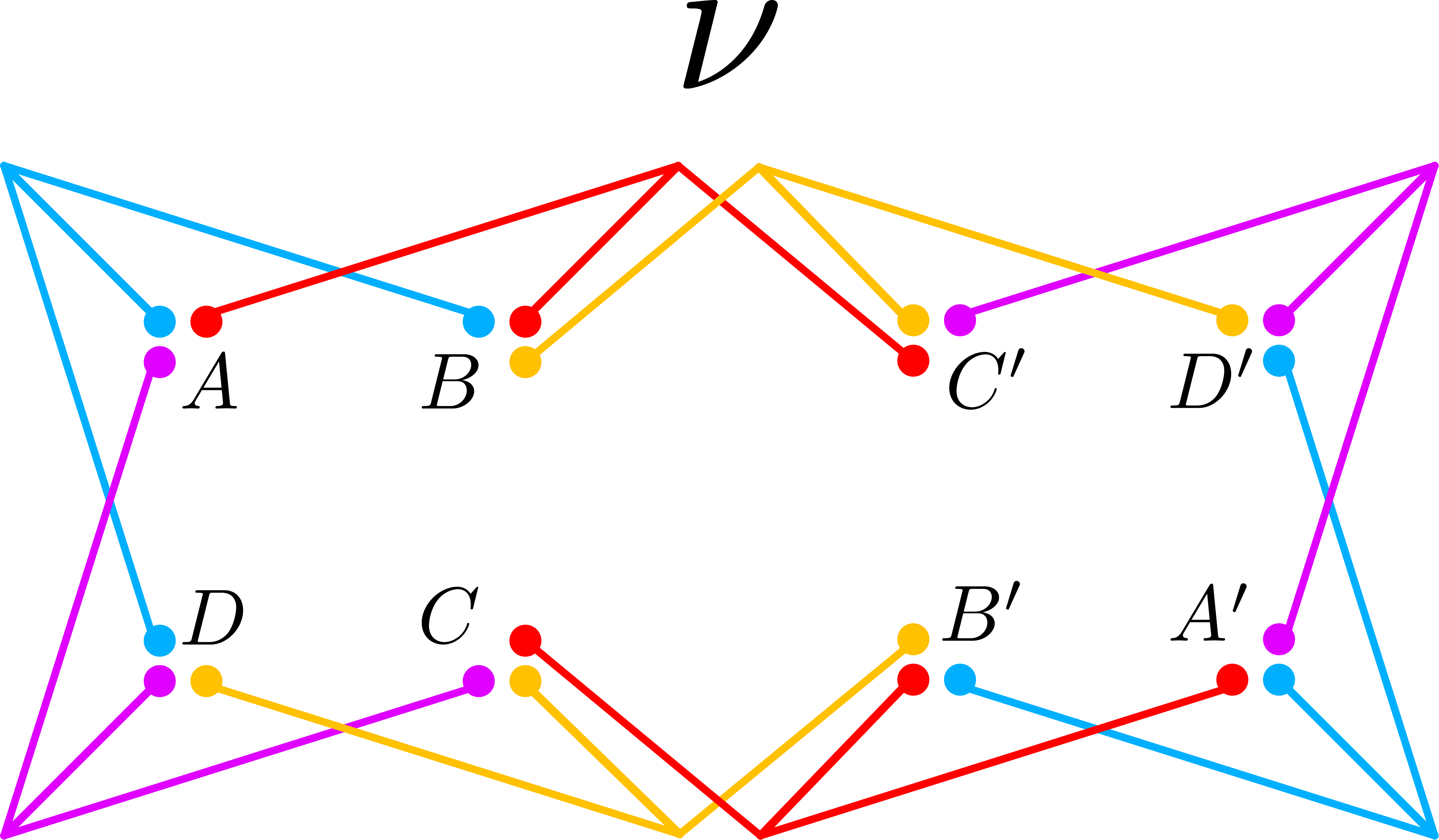}
    \caption{The $\tau$- and $\nu$-inflations of the square network with tripartite sources.}\label{fig:tripartsquareINFL}
\end{figure}%%%%%%%%%%%%%%%%%%%%%%%%%%%%%%%%%%%%%%%%%%%%%%%%%%%%
\begin{figure}[h!]
    \centering
    \includegraphics[height=.2\linewidth]{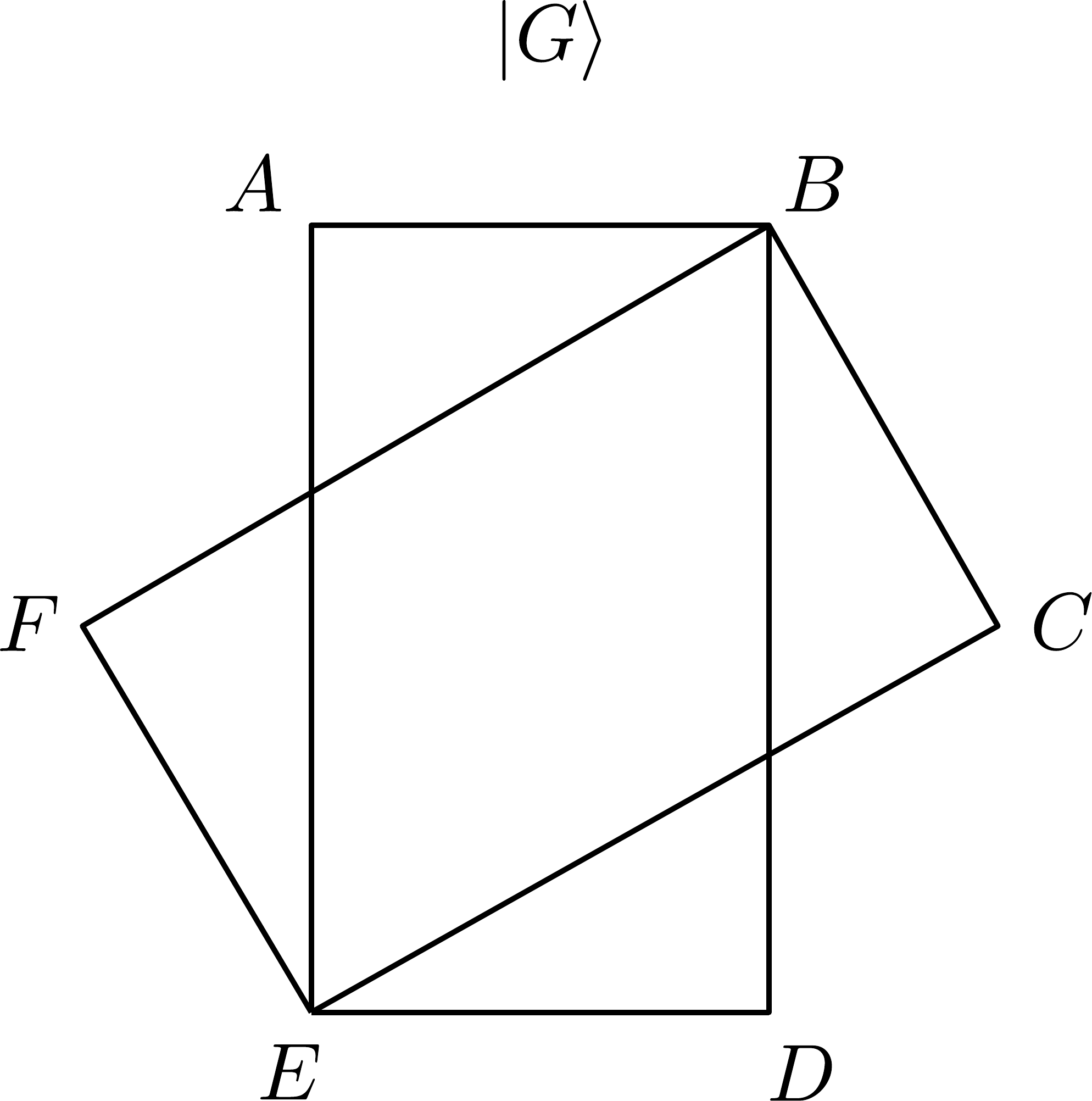}
    \caption{A six-partite graph state $\ket{G}$.}\label{fig:additional2}
\end{figure}
%%%%%%%%%%%%%%%%%%%%%%%%%%%%%%%%%%%%%%%%%%%%%%%%%%%%
\begin{figure}[h!]
    \centering
    \includegraphics[scale=.2]{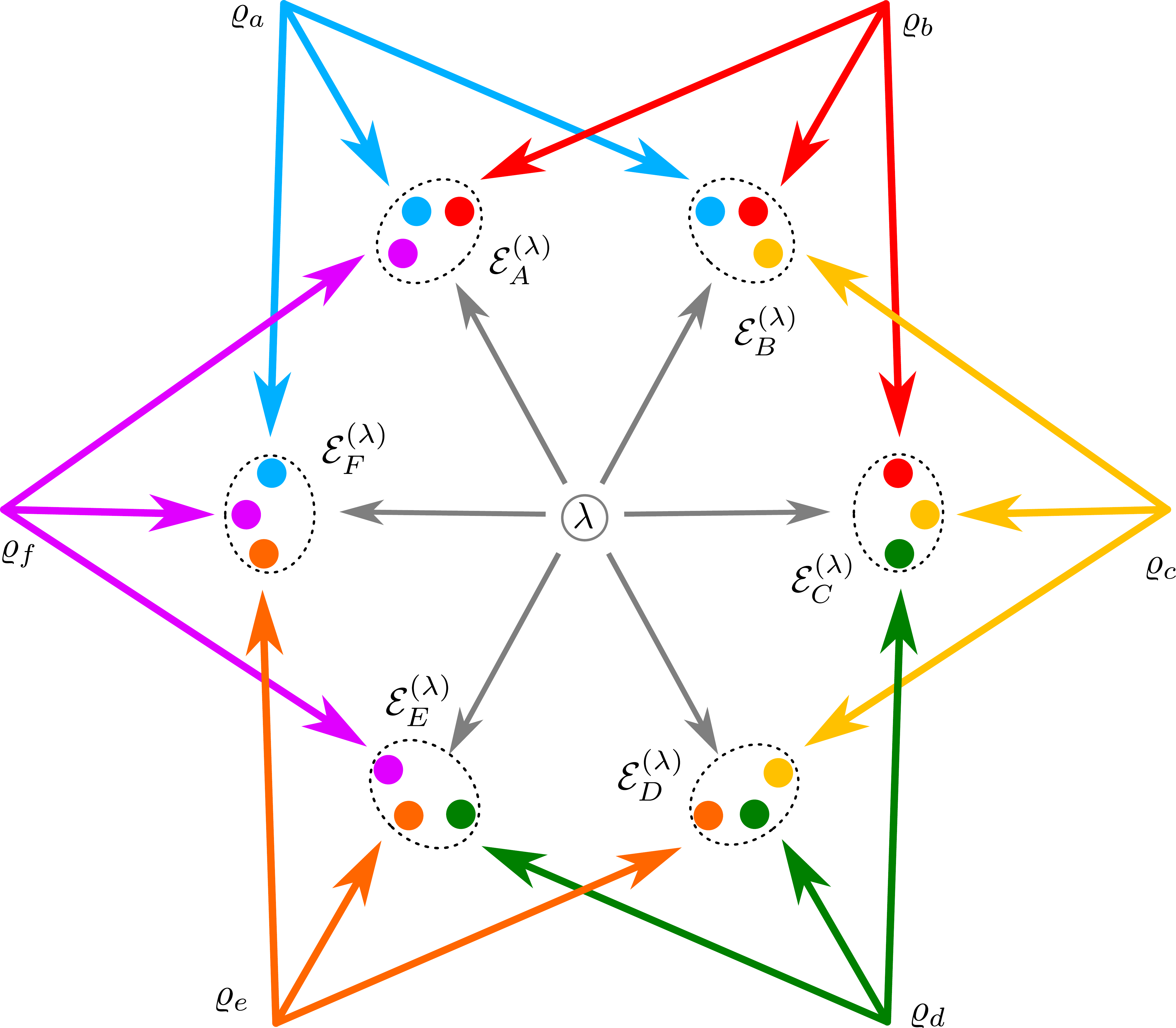}
    \caption{Hexagon network with six tripartite sources.}\label{fig:triparthexa}
\end{figure}
%%%%%%%%%%%%%%%%%%%%%%%%%%%%%%%%%%%%%%%%%%%%%%%%%%%%
\begin{figure}[h!]
    \centering
    \includegraphics[scale=.1]{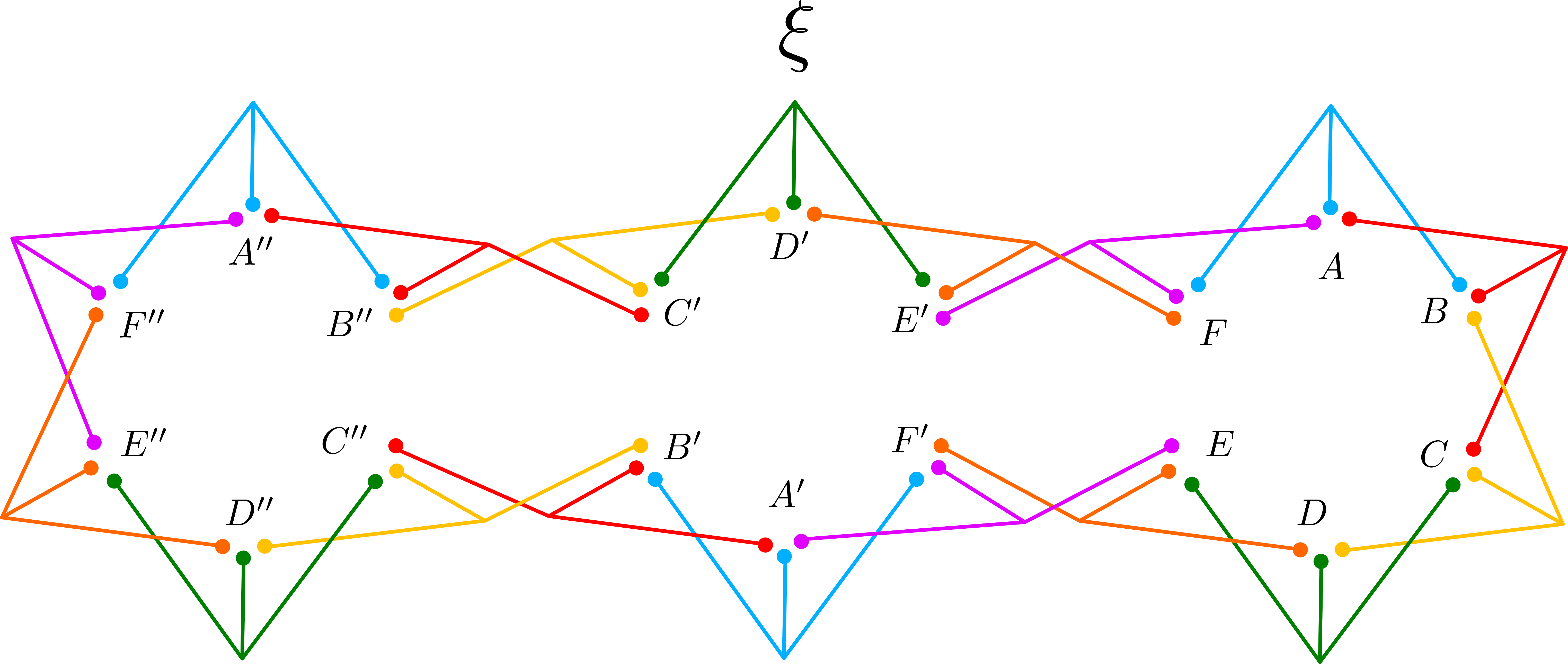}
    \caption{The $\xi$-inflation of the hexagon network with tripartite sources.}\label{fig:triparthexaINLF}
\end{figure}
%%%%%%%%%%%%%%%%%%%%%%%%%%%%%%%%%%%%%%%%%%%%%%%%%%%%

From these two examples, we see that the  anticommuting relation method also holds for networks with more then bipartite sources. However, it remains an open question whether it leads to general results similar to those of Observation 1.

%%%%%%%%%%%%%%%%%%%%%%%%%%%%%%%%%%%%%%%%%%%%%%%%%%%%%%%%%%%%%%
\section{Supplementary Note 4: Permutationally symmetric states}
%\tk{Appendix E}
Before proving our main results, let us give some definitions. As 
introduced in the main text, we define $N$-partite permutationally 
symmetric (bosonic) states as states that satisfy $\Pi^+_{ij} \varrho \Pi^+_{ij}=\varrho$ for all $i,j \in \{1,\dots,N\}$ with $2\Pi^+_{ij} = \id + F_{ij}$ and $F_{ij}$ being the flip operator that exchanges parties $i$ and $j$. We can 
also introduce fermionic states that are antisymmetric for a pair of parties 
$\{ij\}$, i.e.\@ $\Pi^-_{ij} \varrho \Pi^-_{ij}=\varrho$, with $2\Pi^-_{ij} = \id - F_{ij}$. 
For our discussion we need several basic facts. We stress that the 
following Lemma \ref{lemma:SymmStatesDec} and \ref{lemma:GME} are 
well known~\cite{Eckert2002,Kraus2003,Toth2009}, while Lemma 
\ref{lemma:MarginalPerm} is a simple technical statement.

%%%%%%%%%%%%%%%%%%%%%%%%%%%%%%%%%%%%%%%%%%%%%%
\begin{lemma} \label{lemma:SymmStatesDec}
    Let $\varrho=\sum_k p_k \ketbra{\psi_k}$ be a multipartite state 
    and let $\Pi$ be a projector such that $\Pi\varrho \Pi = \varrho$. 
    Then $\Pi \ket{\psi_k} = \ket{\psi_k}$ for all $k$.
\end{lemma}
%%%%%%%%%%%%%%%%%%%%%%%%%%%%%%%%%%%%%%%%%%%%%
\begin{proof}
    One has
    \begin{equation}
        \begin{split}
            1 = \tr(\varrho)    = & \tr(\Pi \varrho \Pi) \\
                            = & \sum_k p_k \bra{\psi_k}\Pi\ket{\psi_k}.
        \end{split}
    \end{equation}
    So $\bra{\psi_k}\Pi\ket{\psi_k} = 1$, and since $\Pi$ is a projector, $ \Pi\ket{\psi_k} =\ket{\psi_k}$, for all $k$.
\end{proof}
This holds in particular for $\Pi=\Pi_{ij}^\pm$. As a second lemma, 
we have
%%%%%%%%%%%%%%%%%%%%%%%%%%%%%%%%%%%%%%%%%%%%%%%%5
\begin{lemma} \label{lemma:MarginalPerm}
	If the reduced state on $AB$ of some state is symmetric or antisymmetric 
	under the exchange of parties $A$ and $B$, then the global state also is.
\end{lemma}
%%%%%%%%%%%%%%%%%%%%%%%%%%%%%%%%%%%%%%%%%%%%%%%
\begin{proof}
	Let $ \varrho = \sum_k p_k \ket{\psi_k}\bra{\psi_k} $ be the state of some tripartite system $ABC$. Let us prove that if $F_{AB} \left(\tr_C (\varrho)\right) = \pm \tr_C (\varrho)$, then $(F_{AB} \otimes \id) \varrho = \pm \varrho$. If one considers the Schmidt decomposition of $\ket{\psi_k}$ wrt the bipartition $AB|C$, one has
	\begin{equation}
		\varrho = \sum_k p_k \sum_{i,j} s_{k,i}s_{k,j}^* \ket{\phi_{k,i}^{AB}}\bra{\phi_{k,j}^{AB}} \otimes \ket{\chi_{k,i}^{C}}  \bra{\chi_{k,j}^{C}}.
	\end{equation}
	From that, $\varrho_{AB} = \sum_k p_k \sum_i |s_{k,i}|^2 \ket{\phi_{k,i}^{AB}} \bra{{\phi_{k,i}^{AB}}} $ and since it is a permutationally symmetric (respectively antisymmetric) state, from Lemma \ref{lemma:SymmStatesDec} all states in its decomposition also are and thus $(F_{AB} \otimes \id) \varrho = \pm \varrho$.
\end{proof}
We note that for both those lemma, the converse is trivial. Finally, we have:

%%%%%%%%%%%%%%%%%%%%%%%%%%%%%%%%%%%%%%%%%%%

\begin{lemma}\label{lemma:GME}
    (a) A $N$-partite symmetric state $\varrho_s$ is either genuinely $N$-partite entangled or fully separable. (b) A $N$-partite antisymmetric state is always $N$-partite entangled. 
\end{lemma}

%%%%%%%%%%%%%%%%%%%%%%%%%%%%%%%%%%%%%%%%%%
\begin{proof}
Due to Lemma \ref{lemma:SymmStatesDec} we only need to consider pure states. 
Let $\ket{\Psi}$ be a $N$-partite (anti)symmetric state that is not $N$-partite 
entangled, hence it is separable for some bipartition. Without loss of 
generality, we assume that
\begin{equation}
\ket{\Psi} = \ket{\varphi_{1,\dots,t}} \otimes \ket{\phi_{t+1,\dots,N}}.
\end{equation}
%Then, by symmetry $\ket{\Psi}$ is also a product for all other 
%$t|(N-t)$ partitions. But, if the state $\ket{\varphi_{1,\dots,t}}$
%contains some entanglement, this directly leads to a contradiction, 
%as then some other $t|(N-t)$ bipartition must be entangled. 
Thus, by tracing out the first $t$ parties, we have a pure state. The symmetry of $\ket{\Psi}$ implies that the reduced state is pure after tracing out any $t$ parties. This can only be true if $\ket{\varphi_{1,\ldots,t}}, \ket{\phi_{t+1,\ldots,N}}$ are fully separable.

Besides, denote $\ket{ab\ldots c}$ a normalized fully separable antisymmetric state, we have $\ket{ab\ldots c} = - \ket{ba\ldots c}$. This implies that
$-1 = \mean{ab\ldots c|ba\ldots c} = |\mean{a|b}|^2 \ge 0$, hence we arrive at a contradiction.
\end{proof}

We note that the notions of entanglement used in this Lemma are the standard ones
for non-symmetric states, as these are the relevant ones for the main text. In principle,  for indistinguihable particles one may separate the ``formal" 
entanglement due to the wave function symmetrization from the ``physical" entanglement \cite{Eckert2002}.

Now, let us prove the Observation 2 of the main text. For completeneness, we restate
it here in the full formulation:

\noindent
{\bf Observation 2'.} {\it Let $\varrho$ be a permutationally symmetric
multiparticle state. Then, $\varrho$ can be generated in a quantum network
with $N-1$-partite sources if and only if it is fully separable.  
If $\varrho$ be a permutationally antisymmetric, then it cannot be generated in a network.}

\begin{proof}
  Let $\varrho$ be a $N$-partite permutationally (anti)symmetric state. Let us assume that it can be generated in a 
network of $N$ nodes with some at most $(N-1)$-partite sources. Note that any 
state that can be generated in a network of $N$ nodes with no $N$-partite 
sources, can also be generated 
in a network of $N$ nodes with $N$ different $(N-1)$-partite sources. Let 
us denote by $\varsigma_{i}$ the source used to generate $\varrho$ that distributes parties to all 
nodes except the $i$th one. 
	
If we assume that $\varrho$ is a network state, then 
the inflation $\eta$ build the following way is a physical state: 
Consider a network of $2N$ nodes $\{A_i, A'_i:i=1,\dots,N\}$ and $2N$ 
sources $\{\zeta_{k}, \zeta'_{k}:k=1,\dots,N\}$ that distribute parties 
to
\begin{align}
  &\zeta_k : A_1\ldots A_{k-1} A'_{k+1}\ldots A'_N,\\ 
  &\zeta'_k : A'_1\ldots A'_{k-1} A_{k+1}\ldots A_N, 
\end{align}
where $\zeta_{k} = \zeta'_{k} = \varsigma_{k}$ for all $k$. The state $\eta$ is the network state build with these sources and the same channels on the nodes than $\varrho$ (with some shared randomness). From the inflation technique, we know that for the reduced states
\begin{align}\label{eq:chainsym}
  &\eta_{A_iA_{i+1}} = \eta_{A'_iA'_{i+1}} = \varrho_{A_iA_{i+1}}, \quad \forall i\le N-1,\\
  &\eta_{A_1A'_N} = \eta_{A'_1A_N} = \varrho_{A_1A_N}.
\end{align}
Since the state $\varrho$ is fully (anti)symmetric, Lemma~\ref{lemma:MarginalPerm} and Eq.~\eqref{eq:chainsym} imply that the state $\eta$ is also fully (anti)symmetric.
	
Now, we consider the inflated state $\tau$, whose sources $\{\omega_{k}, \omega'_{k}: k=1, \dots, N\}$ distribute states to 
\begin{equation}
\begin{split}
  & \omega_{k} : A_1A_2A_3 \dots A_{k-1}A_{k+1}\dots A_N \\
  & \omega'_{k} : A'_1A'_2A'_3 \dots A'_{k-1}A'_{k+1}\dots A'_N.
\end{split}
\end{equation} 
Again, the local channels and shared randomness are the same than for $\varrho$.
This is the two-copy inflation considered several times in the main text. One 
has
\begin{equation}
  \tau_{A_1\dots A_N} = \tau_{A'_1\dots A'_N} = \varrho.
\end{equation}
Moreover,
\begin{equation}
  \tau_{A_iA'_i} = \eta_{A_iA'_i} 
\end{equation}
hence $\tau$ is permutationally fully (anti)symmetric under the exchange of all its parties. However, $\tau$ is separable wrt the bipartition $A_1\dots A_N | A'_1\dots A'_N$. In the fully symmetric case, this means that $\tau$ is fully separable. Therefore $\varrho$ is also fully separable. So, if a network state is permutationally symmetric, it needs to be fully separable. In the fully antisymmetric case, the full separability of $\tau$ contradicts with the assumption that $\tau$ is fully antisymmetric. So, no network state can be permutationally antisymmetric.
\end{proof}

Finally, we note that cyclic symmetric states may be generated in network scenarios, as already pointed out by Ref.~\cite{Luo2020}. As an example, we consider three Bell pairs $\ket{\Phi^+}$ as sources in the triangle network, and no channels applied. The global $3 \times 4$-partite state is $\ket{\Psi}_{ABC} = \ket{\Phi^+}_{A_2B_1} \ket{\Phi^+}_{B_2C_1} \ket{\Phi^+}_{C_2A_1}$, with $A=A_1A_2$ and so on. With the appropriate reordering of the parties and by mapping $\ket{ij}_X \mapsto \ket{2i+j}_X$ for $X=A,B,C$, one gets
\begin{equation}
   % \begin{split}
        \ket{\Psi}_{ABC} = %& 
        \frac{1}{2\sqrt{2}} \big( \ket{000} + \ket{012} + \ket{120} + \ket{201} %\\ & 
        + \ket{132} + \ket{321} + \ket{213} + \ket{333} \big),
%    \end{split}
\end{equation}
which is a symmetric state under cyclic permutations of the $3 \times 4$-dimensional system $ABC$.

%%%%%%%%%%%%%%%%%%%%%%%%%%%%%%%%%%%%%%%%%%%%%%%%%%%%%%%%%%%%5
\section{Supplementary Note 5: Certifying network links}

%%%%%%%%%%%%%%%%%%%%%%%%%%%%%%%%%%%%%%%%%%%%%%%%%%%%%%%%%%%%%
Here, we prove the statement made in the main text, which can be formulated
as follows:

%%%%%%%%%%%%%%%%%%%%%%%%%%%%%%%%%%%%%%%%%%%%%%%%%%%%%%%%%%%%%%%%%%%%%%%%5
\begin{observation}
If a state $\varrho$ can be prepared in a network with bipartite sources but 
without the link $AC$, then
\begin{equation}
\mean{X_A X_C P_{R_1}}^2 
+ 
\mean{Y_A Y_C P_{R_2}}^2 
+ 
\mean{Z_A Z_C P_{R_3}}^2 
\leq 1.
\label{eq-linktest-app}
\end{equation}
Here the $P_{R_i}$ are arbitrary observables on disjoint subsets of the other
particles, $R_i \cap R_j = \emptyset$. If the state $\varrho$ was indeed  prepared 
in a real quantum network, then violation of this inequality proves that the link 
$AC$ is working and distributing entanglement.
\end{observation}
%%%%%%%%%%%%%%%%%%%%%%%%%%%%%%%%%%%%%%%%%%%%%%%%%%%%%%%%%%%%%%%%%
\begin{proof}
Without loss of generality, we assume that
\begin{equation}
    R_1 = \{E\}, \quad R_2 = \{B\}, \quad R_3 = \{D\}.
\end{equation}
Otherwise, we can prove the result similarly. The disconnected nodes $A$
and $C$ may be connected via some source with the $R_i$ or not, but this 
is not essential. Then, the graph has a structure as the graph in Supplementary Fig.~\ref{fig:trianglerho2}.

We consider a three-copy inflation $\xi$ of this graph, where the observables in 
Eq.~(\ref{eq-linktest-app}) overlap only in the node $A$. This inflation is
constructed as follows: All links from $B$ to $C$ are rewired from $B$ to $C'$
and all links from  $D$ to $C$ are rewired from $D$ to $C''$. This is shown 
schematically in Supplementary Fig.~\ref{fig:trianglerho2}. 

%%%%%%%%%%%%%%%%%%%%%%%%%%%%%%%%%%%%%%%%%%%%%%%%%%%%%%%%%%%%%%%%55
\begin{figure}[h!] 
    \centering
    \includegraphics[scale=.2]{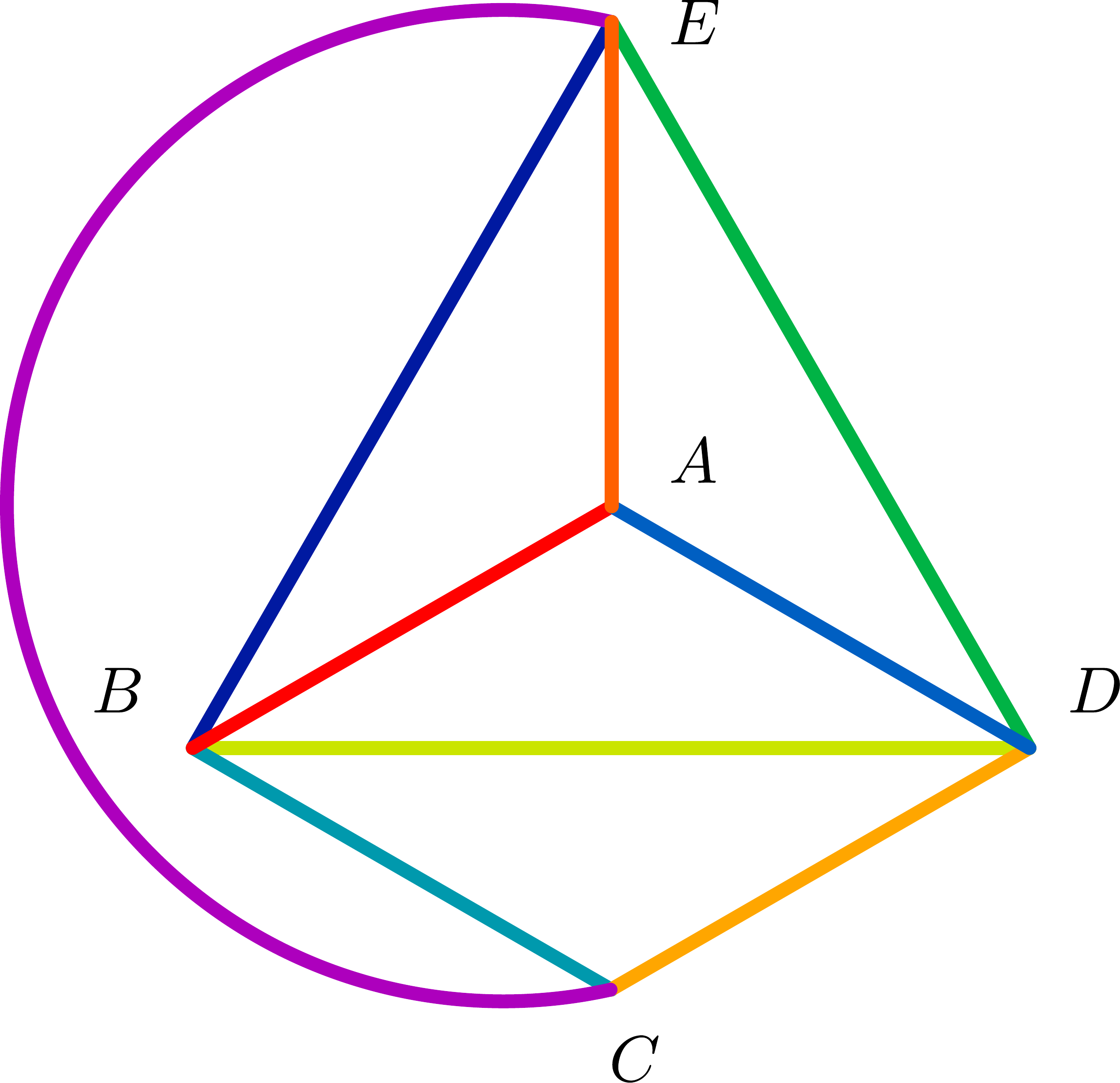}
    \hspace{1cm}
    \includegraphics[scale=.2]{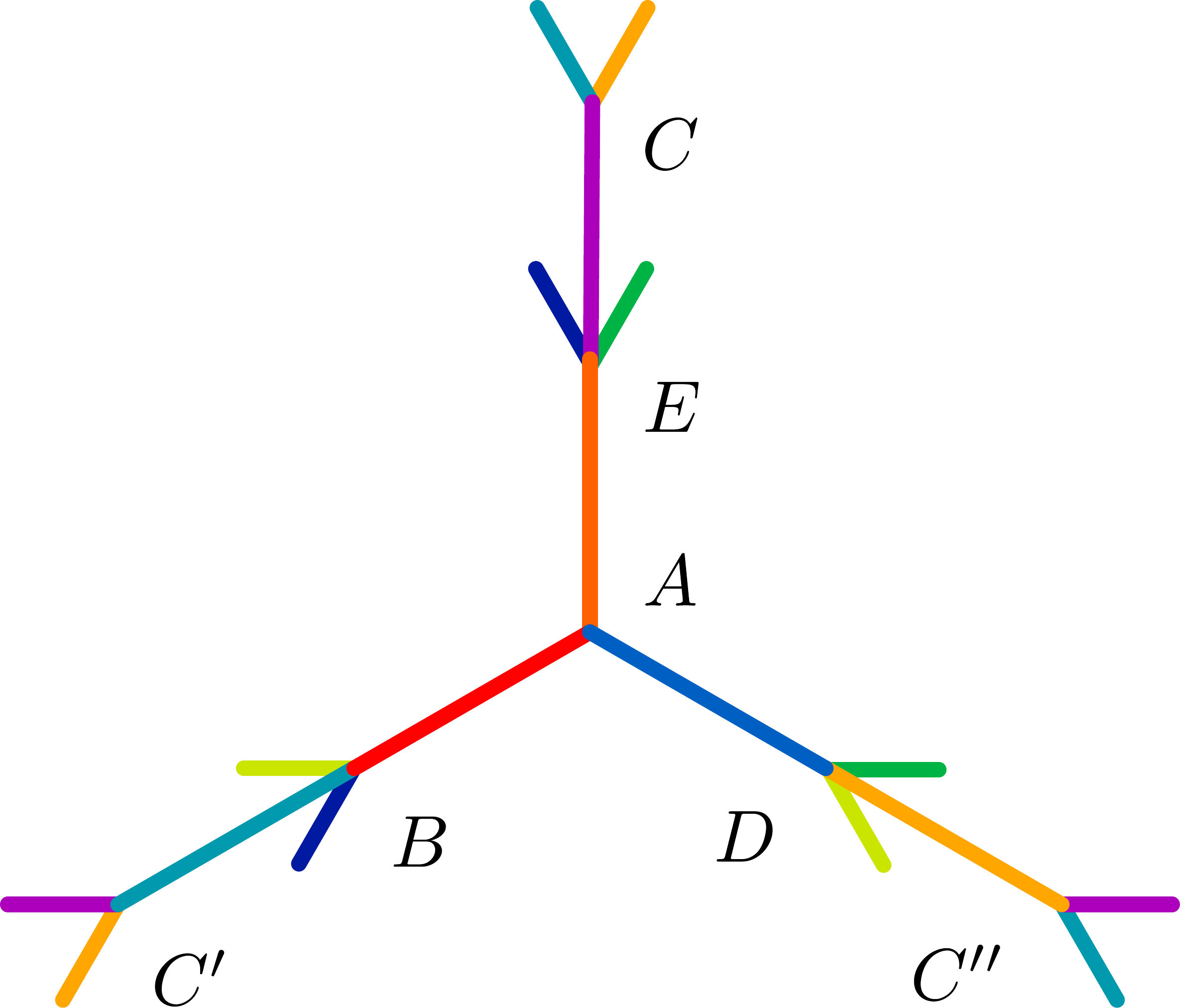}
    \caption{\textit{Left: }A typical graph for a network where the link $AC$ is missing. \textit{Right: }Schematic view of the third order inflation. See text for further details.
    }
    \label{fig:trianglerho2}
\end{figure}
%%%%%%%%%%%%%%%%%%%%%%%%%%%%%%%%%%%%%%%%%%%%%%%%%%%%%%%%%%%%%%%%%%%%%%%%
The anticommuting relations imply that
\begin{equation}\label{eq:meantri}
    \mean{X_AX_CP_E}_\xi^2 + \mean{Y_AY_{C'}P_B}_\xi^2 + \mean{Z_AZ_{C''}P_D}_\xi^2 \le 1.
\end{equation}
By comparing the marginals of $\varrho $ and $\xi$, we have
\begin{align}
    &\mean{X_AX_CP_E}_\xi = \mean{X_AX_CP_E}_\varrho,\\
    &\mean{Y_AY_CP_B}_\xi = \mean{Y_AY_CP_B}_\varrho,\\
    &\mean{Z_AZ_CP_D}_\xi = \mean{Z_AZ_CP_D}_\varrho.
\end{align}
By substituting the mean values with state $\xi$ by the ones with $\varrho$ in Eq.~\eqref{eq:meantri}, we complete our proof.
\end{proof}

\muell{
%%%%%%%%%%%%%%%%%%%5
\begin{figure} 
    \centering
    \includegraphics[scale=.2]{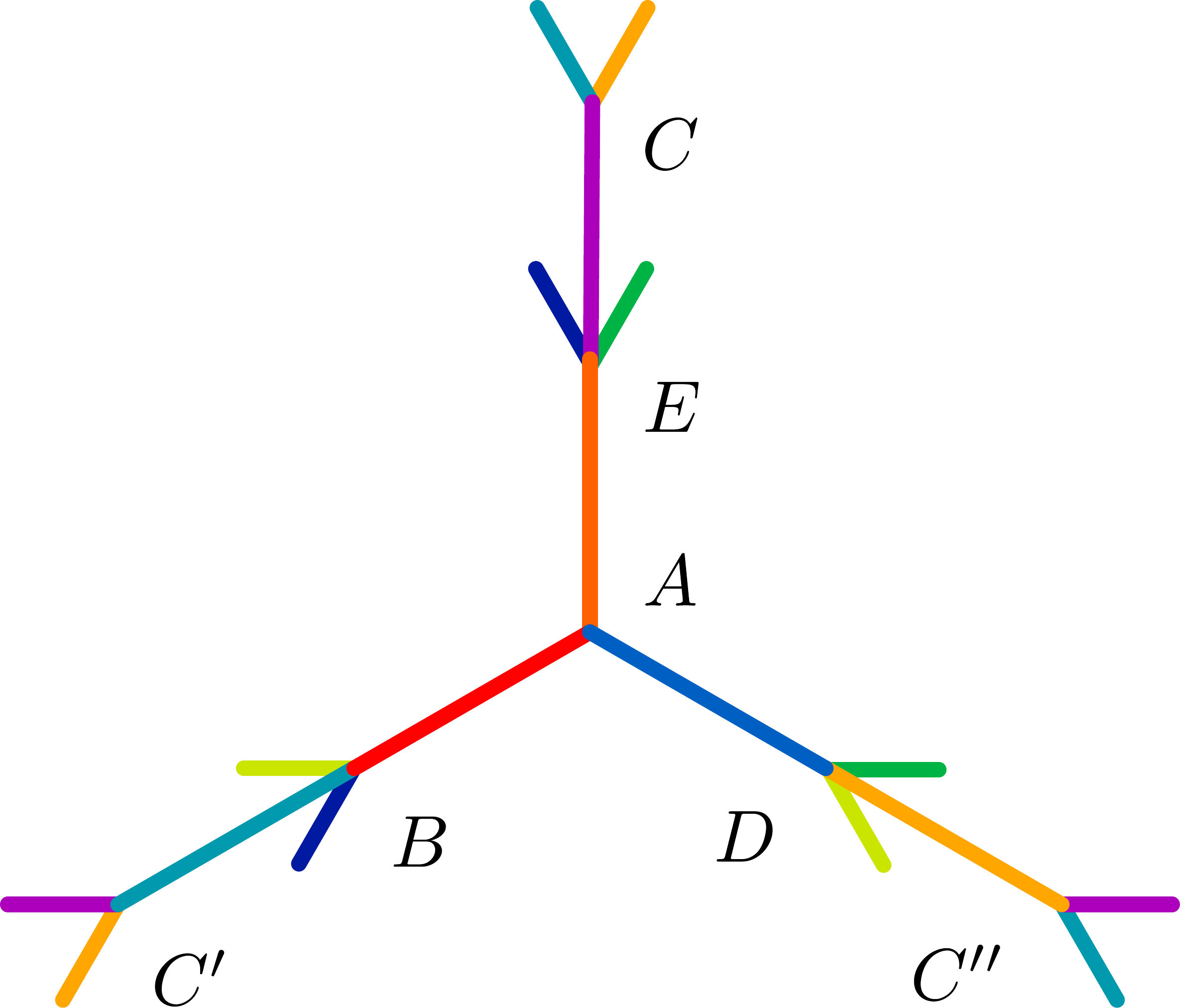}
    \caption{Schematic view of the third order inflation $\xi$ of the network in 
    Supplementary Fig.~\ref{fig:trianglerho2}. See text for further details.}
    \label{fig:trianglelambda}
\end{figure}
%%%%%%%%%%%%%%%%%%%%%%%%%%%%%%%%%%%%%%%%%%%%%%
}

Finally, we give a simple example where this criterion
can detect the functionality of a link, while a simple
concentration on the reduced two-qubit density matrix does
not work. Consider the state
\begin{equation}
    \varrho = \frac{1}{2} (\ket{s_1}\bra{s_1} + \ket{s_2}\bra{s_2}),
\end{equation}
where 
\begin{align}
    &\ket{s_1} = (\ket{00}_{AC} + \ket{11}_{AC})\otimes \ket{00}_{BD},\\
    &\ket{s_2} = (\ket{00}_{AC} - \ket{11}_{AC})\otimes \ket{11}_{BD}.
\end{align}
Here, we want to check whether the link $AC$ works or not. Since the reduced state on $AC$ is separable, we cannot use the criteria which acts only on $AC$. It is easy to verify that 
\begin{equation}
    \begin{split}
         &\mean{X_A \id_{B} X_C \id_{D}} = 0, \\ -&\mean{Y_A \id_{B} Y_C Z_D} = \mean{Z_A \id_{B} Z_C \id_{D}} = 1,
    \end{split}
\end{equation}
which violates Eq.~\eqref{eq:meantri}. Hence, our criteria can detect the link 
$AC$ more effectively.

\twocolumngrid

\end{document}